\documentclass[11pt]{article}
\def\x{2.9}
\usepackage[lmargin=\x cm,rmargin=\x cm,tmargin=\x cm,bmargin=\x cm,bindingoffset=0cm]{geometry}
\usepackage{amsthm}
\usepackage{nicefrac}
\usepackage{amsmath}
\usepackage{amssymb}
\usepackage{amsfonts}
\usepackage{mathrsfs}
\usepackage{tikz}
\usetikzlibrary{cd}
\usepackage{url}
\usepackage{color}  
\usepackage{mathscinet}

\newcommand{\mc}[1]{\mathcal{#1}}
\newcommand{\mb}[1]{\mathbb{#1}}
\newcommand{\mf}[1]{\mathfrak{#1}}

\newcommand{\bigslant}[2]{{\raisebox{.1em}{$#1$}\!\!\left/\raisebox{-.1em}{$#2$}\right.}}

\newcommand{\rey}[1]{\mathcal{R}_{#1}}
\newcommand{\alia}{\mf{A}}

\newcommand{\roots}{\Phi}

\newcommand{\splitk}{\mb{C}}
\newcommand{\triv}{\epsilon}
\newcommand{\natrep}{U}
\newcommand\regchar[1]{\chi_{\splitk{#1}}}

\newcommand{\rs}{\mb{X}}
\newcommand{\cp}{\mb{CP}^1}
\newcommand{\mero}[1][\cp]{\mathcal{M}(#1)}
\newcommand{\eval}[1]{\mathrm{ev}_{#1}}
\renewcommand{\div}[1]{\text{div}\!\left(#1\right)}
\newcommand{\ord}[2]{\text{ord}_{#1}\!\left(#2\right)}

\newcommand{\Id}{\mathrm{Id}}
\newcommand{\id}{\mathrm{id}}
\newcommand{\tr}{\mathrm{tr}\,}

\newcommand{\GL}{\mathrm{GL}}
\newcommand{\SL}{\mathrm{SL}}

\newcommand{\PSL}{\mathrm{PSL}}

\newcommand{\PGL}{\mathrm{PGL}}

\newcommand{\Aut}[1]{\mathrm{Aut}\!\left(#1 \right)}
\newcommand{\Int}[1]{\mathrm{Int}\!\left(#1 \right)}

\newcommand{\ad}{\mathrm{ad}}

\newcommand{\lcm}{\mathrm{\,lcm}}
\newcommand{\codim}{\mathrm{codim}\,}
\newcommand{\diag}{\mathrm{diag}}
\newcommand{\rank}{\mathrm{rank}\,}

\newcommand{\al}{{a}}
\newcommand{\be}{{b}}
\newcommand{\ga}{{c}}
\newcommand{\fal}{F_\al}
\newcommand{\fbe}{F_\be}
\newcommand{\fga}{F_\ga}

\newcommand{\tal}{{\theta_\al}}
\newcommand{\tbe}{{\theta_\be}}
\newcommand{\tga}{{\theta_\ga}}
\newcommand{\ti}{{\theta_i}}

\newcommand{\nal}{{\nu_\al}}
\newcommand{\nbe}{{\nu_\be}}
\newcommand{\nga}{{\nu_\ga}}
\newcommand{\dal}{d_\al}
\newcommand{\dbe}{d_\be}
\newcommand{\dga}{d_\ga}
\newcommand{\Gal}{\Gamma_\al}
\newcommand{\Gbe}{\Gamma_\be}
\newcommand{\Gga}{\Gamma_\ga}
\newcommand{\ii}{\mb{I}_i}
\newcommand{\ial}{\mb{I}_\al}
\newcommand{\ibe}{\mb{I}_\be}
\newcommand{\iga}{\mb{I}_\ga}

\newcommand{\finitegroupfont}{\mathsf}
\newcommand{\zn}[1]{\mb{Z}/#1\mb{Z}}
\newcommand{\cg}[1]{\finitegroupfont{C}_{#1}}
\newcommand{\dg}[1]{\finitegroupfont{D}_{#1}}
\newcommand{\tg}{\finitegroupfont{T}}
\newcommand{\og}{\finitegroupfont{O}}
\newcommand{\yg}{\finitegroupfont{Y}}

\newcommand{\bt}{\finitegroupfont{T}^{\flat}}
\newcommand{\bo}{\finitegroupfont{O}^{\flat}}
\newcommand{\by}{\finitegroupfont{Y}^{\flat}}
\newcommand{\bti}{\finitegroupfont{T}_{1}}
\newcommand{\btii}{\finitegroupfont{T}_{2}}
\newcommand{\btiii}{\finitegroupfont{T}_{3}}
\newcommand{\btiiii}{\finitegroupfont{T}^{\flat}_{4}}
\newcommand{\btiiiii}{\finitegroupfont{T}^{\flat}_{5}}
\newcommand{\btiiiiii}{\finitegroupfont{T}^{\flat}_{6}}
\newcommand{\btiiiiiii}{\finitegroupfont{T}_{7}}

\newcommand{\boi}{\finitegroupfont{O}_{1}}
\newcommand{\boii}{\finitegroupfont{O}_{2}}
\newcommand{\boiii}{\finitegroupfont{O}_{3}}
\newcommand{\boiiii}{\finitegroupfont{O}^{\flat}_{4}}
\newcommand{\boiiiii}{\finitegroupfont{O}^{\flat}_{5}}
\newcommand{\boiiiiii}{\finitegroupfont{O}_{6}}
\newcommand{\boiiiiiii}{\finitegroupfont{O}_{7}}
\newcommand{\boiiiiiiii}{\finitegroupfont{O}^{\flat}_{8}}

\newcommand{\byi}{\finitegroupfont{Y}_{1}}
\newcommand{\byii}{\finitegroupfont{Y}^{\flat}_{2}}
\newcommand{\byiii}{\finitegroupfont{Y}^{\flat}_{3}}
\newcommand{\byiiii}{\finitegroupfont{Y}_{4}}
\newcommand{\byiiiii}{\finitegroupfont{Y}_{5}}
\newcommand{\byiiiiii}{\finitegroupfont{Y}_{6}}
\newcommand{\byiiiiiii}{\finitegroupfont{Y}^{\flat}_{7}}
\newcommand{\byiiiiiiii}{\finitegroupfont{Y}_{8}}
\newcommand{\byiiiiiiiii}{\finitegroupfont{Y}^{\flat}_{9}}

\newtheorem{Theorem}{Theorem}[section]
\newtheorem{Proposition}[Theorem]{Proposition}
\newtheorem{Lemma}[Theorem]{Lemma}
\newtheorem{Corollary}[Theorem]{Corollary}

\newtheorem{Definition}[Theorem]{Definition}
\newtheorem{Example}[Theorem]{Example}

\newtheorem{Remark}[Theorem]{Remark}

\title{Hereditary Automorphic Lie Algebras}
\date{}

\author{Vincent Knibbeler$^{a,b}$, Sara Lombardo$^{a}$ and Jan A. Sanders$^{b}$ \\
\\
$^{a}$Department of Mathematical Sciences, School of Science\\
Loughborough University, Schofield Building, LE11 3TU, Loughborough, UK\\
\and
\\
$^{b}$Department of Mathematics, Faculty of Sciences\\
Vrije Universiteit, De Boelelaan 1081a, 1081 HV Amsterdam, The Netherlands
}

\begin{document}
\maketitle
\begin{abstract}
We show that  Automorphic Lie Algebras which contain a Cartan subalgebra with a constant spectrum, called \emph{hereditary}, are completely described by 2-cocycles on a classical root system taking only two different values. This observation suggests a novel approach to their classification. By determining the values of the cocycles on opposite roots, we obtain the Killing form and the abelianisation of the Automorphic Lie Algebra.
The results are obtained by studying equivariant vectors on the projective line. As a byproduct, we describe a method to reduce the computation of the infinite dimensional space of said equivariant vectors to a finite dimensional linear computation and the determination of the ring of automorphic functions on the projective line.
\vspace{1cm}

\noindent\textbf{AMS  Subject Classification Numbers}
\\13A50 (commutative algebra: general commutative ring theory: actions of groups on commutative rings; invariant theory); 
\\17B05 (nonassociative rings and algebras: Lie algebras and Lie superalgebras: structure theory);
\\17B65 (nonassociative rings and algebras: Lie algebras and Lie superalgebras: infinite-dimensional Lie (super)algebras);
\\17B80 (nonassociative rings and algebras: Lie algebras and Lie superalgebras: applications to integrable systems);

\vspace{1cm}
\noindent\textbf{Key words}\\
Rational Equivariant Vectors, Automorphic Lie Algebras

\vspace{1cm}
The second author gratefully acknowledges financial support from EPSRC (EP/E044646/1 and EP/E044646/2) and from NWO VENI (016.073.026).
\end{abstract}

\section{Introduction}
Automorphic Lie Algebras \cite{lombardo2004reductions,lombardo2005reduction} 
are 
algebras of meromorphic maps from the projective line $\cp$ to a simple complex Lie algebra $\mf{g}$ which are equivariant with respect to a discrete group $G$ embedded in $\Aut{\cp}$ and $\Aut{\mf{g}}$. 
Furthermore, the poles of these maps are confined to a ($G$-invariant) set $\Gamma\subset\cp$. In a more abstract context, Automorphic Lie Algebras are also known 
as Equivariant Map Algebras \cite{neher2012irreducible}, the latter enjoying a more general definition.
They were introduced as a tool for the classification of integrable partial differential equations in $1+1$ dimensions. However, being generalisations of the remarkably fruitful theory of (twisted) loop algebras, there are many more application opportunities for Automorphic Lie Algebras. 
Moreover, the initial classification results reveal isomorphisms between Automorphic Lie Algebras which were unexpected, 
despite the set up involving only classical areas of mathematics, making an explanation interesting within Lie theory.

The work reported in this paper is motivated by the desire to understand patterns revealed by explicit computations of these algebras. Indeed, the
computational programme \cite{knibbeler2017higher} unveiled the Lie structure of various Automorphic Lie Algebras, which was, to a large extent, determined by integers $\kappa_\al$, $\kappa_\be$ and $\kappa_\ga$ counting the automorphic functions $\ial$, $\ibe$ and $\iga$ appearing in the structure constants (in a suitable basis). Here $\ial$, $\ibe$ and $\iga$ are the primitive automorphic functions that vanish on the orbits with nontrivial isotropy in the reduction group. 
Collecting these numbers lead to the surprising observation that they are independent of the finite group used to define the Automorphic Lie Algebra (at least, for all the algebras computed at the time).
Meanwhile, investigating Lie algebras of the form $\mf{g}^\theta=\{x\in\mf{g}\,|\,\theta x=x\}$, where $\mf{g}$ is a simple complex Lie algebra and $\theta$ a generator of a polyhedral group $G$ acting on $\mf{g}$, resulted in a list of numbers that was very similar to the former list of numbers $\kappa_\al$, $\kappa_\be$ and $\kappa_\ga$. A formula relating the two was quickly obtained. In this paper we get to the heart of these observations and prove the formula in greater generality.

In order to explain and extend the observed phenomena, we study equivariant vectors on the projective line.
Equivariant vectors play an important role in various branches of mathematics. They are natural objects of study in representation theory and invariant theory, and they received particular attention in the theory of dynamical systems (e.g.~equivariant bifurcations, cf.~\cite{golubitsky1988singularities,golubitsky2002the,golubitsky1985singularities,dias2009hopf,antonelli2008invariants}), as well as in mathematical physics, where rational maps can be used to understand the structure of solutions of the Skyrme model \cite{houghton1998rational}, known as Skyrmions (e.g. \cite{manton2004topological} and references therein).
Obtaining explicit descriptions of the full space of equivariant vectors is usually a substantial computational problem, and various methods have been developed to carry out such computations effectively \cite{gatermann2000computer}. This paper provides information on the space of equivariant vectors that do not require computational effort. In particular, we find the determinant of a basis (Theorem \ref{thm:determinant of invariant vectors}). This determinant of invariant vectors establishes a formula to obtain the aforementioned integers $\kappa_\al$, $\kappa_\be$ and $\kappa_\ga$ in a simple fashion, and thereby explains a significant part of the Lie structure of Automorphic Lie Algebras. But a stronger result is found. Using the description of the evaluation of equivariant vectors (Lemma \ref{lem:evaluating invariant vectors}) from \cite{neher2012irreducible}, we find that an Automorphic Lie Algebra with a constant spectrum CSA, a hereditary Automorphic Lie Algebra hereafter, is described by a $2$-cocycle $\omega^2$ on the root system $\roots$ of $\mf{g}$ (Theorem \ref{thm:alia normal form}). 
The triviality of the second cohomology group \cite{knibbeler2019cohomology} then implies in particular that Automorphic Lie Algebras with poles in all exceptional orbits are nothing but current algebras of $\mf{g}$ (generalising known facts about twisted loop algebras). Including the information of the determinant of invariant vectors yields a partial description of $\omega^2$, enough to describe the Killing form of the Automorphic Lie Algebra and its abelianisation (as well as $\kappa_i=\frac{1}{2}\sum_{\alpha\in\roots}\omega^2(\alpha,-\alpha)$). Moreover, it shows that $\omega^2$ can only take two values, thereby reproducing isomorphisms obtained in \cite{lombardo2004reductions,lombardo2005reduction,lombardo2010on,bury2010automorphic,knibbeler2014automorphic,knibbeler2017higher}.
The challenge remains to find the complete description of $\omega^2$, and hence the Automorphic Lie Algebra, from the embedding of the reduction group into $\Aut{\mf{g}}$.
 
The present paper, together with \cite{knibbeler2019cohomology,knibbeler2014invariants}, is part of a programme to describe and classify Automorphic Lie Algebra theoretically, as opposed to the computational classification programme which is described in \cite{lombardo2004reductions,lombardo2005reduction,lombardo2010on,bury2010automorphic,knibbeler2014automorphic,knibbeler2017higher}. 
The two approaches reinforce each other.
The nature of this paper is more computational than is customary in this area of mathematics, because it goes hand in glove with the actual computation of Automorphic Lie Algebras, which needs these mathematical facts in order to stand on firmer ground.

This paper is organised as follows. In Section \ref{sec:polyhedral groups} we review finite groups of automorphisms of the projective line, and prove a very useful formula for their representations, relating dimensions of subspaces fixed by stabiliser subgroups (Lemma \ref{lem:dimV^g}). In Sections \ref{sec:homogenisation} through \ref{sec:divisor of invariant vectors} we describe equivariant vectors on the projective line. We involve Lie algebras only in Section \ref{sec:alias by cocycles} were the results on equivariant vectors are used to study the structure of Automorphic Lie Algebras. Finally, we include an appendix with several explicit calculations illustrating the theoretical results on equivariant vectors.

\section{Polyhedral Groups}
 \label{sec:polyhedral groups}
A polyhedral group $G$ is a finite group of rotations in $3$-dimensional space. Restricting to a $2$-sphere identified with the Riemann sphere $\overline{\mb{C}}=\mb{C}\cup\{\infty\}$, rotations take the form of M\"obius transformations $\lambda\mapsto \frac{a\lambda+b}{c\lambda+d}$. Another equivalent way to write this is in terms of the complex projective line $\cp$ (the space of equivalence classes $[X:Y]$, where $X$ and $Y$ are complex numbers, not both zero, and $[X_0:Y_0]=[X:Y]$ if and only if there is a nonzero complex number $z$ such that $(zX_0,zY_0)=(X,Y)$). We identify the Riemann sphere with the projective line using the isomorphism given by $\lambda\mapsto[\lambda:1]$, for $\lambda\in\mb{C}$ and $\infty\mapsto [1:0]$. The M\"obius transformations $\lambda\mapsto \frac{a\lambda+b}{c\lambda+d}$ then correspond to $[X:Y]\mapsto [aX+bY:cX+dY]$.
 
Polyhedral groups are well studied and described, cf.~\cite{dolgachev2009mckay,klein1956lectures,klein1993vorlesungen,toth2002finite}.
Consider the action of $G$ on $\cp$. Let $\Omega\ni i$ be an index set for the $G$-orbits $\Gamma_i$ of elements with nontrivial stabiliser groups (nontrivial isotropy). Such orbits will be called \emph{exceptional orbits} hereafter. Moreover, let $d_i=|\Gamma_i|$ be the size of an exceptional orbit and $\nu_i$ the order of an associated stabiliser subgroup. In particular
\[d_i\nu_i=|G|,\qquad i\in\Omega.\]
The classification of polyhedral groups is obtained through the formula 
\begin{equation}
\label{eq:finite subgroups of SO(3)}
2\left(1-\frac{1}{|G|}\right)=\sum_{i\in\Omega}\left(1-\frac{1}{\nu_i}\right).
\end{equation}
This is the genus zero case of the Riemann-Hurwitz formula, but it can be obtained via a more elementary analysis.
From this formula it follows that there are either $2$ or $3$ exceptional orbits. If $|\Omega|=2$ one finds the cyclic groups and if $|\Omega|=3$ one obtains the orientation preserving symmetry groups of the Platonic solids and regular polygons embedded in $\mb{R}^3$. 
Table \ref{tab:various properties of polyhedral groups} contains various relevant numbers for all polyhedral groups, such as their exponent $\|G\|=\min\{n\in\mb{N}\;|\;g^n=1,\,\forall g\in G\}$ and the order of their Schur multiplier $|M(G)|$.
\begin{center}
\begin{table}[ht!] 
\caption{Groups $G$, orders, exponent $\|G\|$, Schur multiplier $M(G)$ and Abelianisation $\mc{A} G$.}
\label{tab:various properties of polyhedral groups}
\begin{center}
\begin{tabular}{llllllll}\hline
$G$&$|\Omega|$&$(\nal, \nbe(,\nga))$&$(\dal, \dbe(,\dga))$&$|G|$&$\|G\|$&$|M(G)|$&$\mc{A} G$\\
\hline
$\cg{N}$&$2$&$(N,N)$&$(1,1)$&$N$&$N$&$1$&$\cg{N}$\\
$\dg{N=2M-1}$&$3$&$(N,2,2)$&$(2,N,N)$&$2N$&$2N$&$1$&$\cg{2}$\\
$\dg{N=2M}$&$3$&$(N,2,2)$&$(2,N,N)$&$2N$&$N$&$2$&$\cg{2}\times \cg{2}$\\
$\tg$&$3$&$(3,3,2)$&$(4,4,6)$&$12$&$6$&$2$&$\cg{3}$\\
$\og$&$3$&$(4,3,2)$&$(6,8,12)$&$24$&$12$&$2$&$\cg{2}$\\
$\yg$&$3$&$(5,3,2)$&$(12,20,30)$&$60$&$30$&$2$&$1$\\
\hline 
\end{tabular}
\end{center}
\end{table}
\end{center}
Most attention will go to the non-cyclic groups, the groups with three exceptional orbits. For convenience we fix in this situation 
\[\Omega=\{\al,\be,\ga\},\]
choose $\nal\ge\nbe\ge\nga$ and use the presentation 
\begin{equation}
\label{eq:presentation of G}
G=\langle \tal,\, \tbe,\, \tga\;|\; \tal^{\nal}=\tbe^{\nbe}=\tga^{\nga}=\tal \tbe \tga=1\rangle.
\end{equation}
Notice the Euler characteristic of the sphere, $\dal-\dga+\dbe=2$.

Representations of polyhedral groups have the following curious property. The special case of the icosahedral group is discussed by Lusztig \cite{lusztig2003homomorphisms} where it is attributed to Serre. The general proof given in \cite{knibbeler2014invariants} is included below.
\begin{Lemma}
\label{lem:dimV^g}
If $G$ is a noncyclic polyhedral group and $V$ a $G$-module, then
\[\dim V+2\dim V^{G}=\sum_{i=\al,\be,\ga}\dim V^{\langle \ti \rangle}\]
where $\ti$ is a generator of a stabilizer subgroup at an exceptional orbit $\Gamma_i$.
\end{Lemma}
\begin{proof}
The proof is based on the observation that every nontrivial rotation of the sphere leaves exactly two points fixed. In other words, all nontrivial group elements of $G<\Aut{\cp}$ occur in precisely two stabiliser subgroups $\{G_\lambda\,|\,\lambda\in\cp\}$. We rewrite the Reynolds operator $\rey{G}$ (or averaging operator) of $G$ by rearranging the terms, based on the mentioned observation;
\begin{align*}
\rey{G}&=\frac{1}{|G|}\sum_{\theta\in G}\theta
\\&=\frac{1}{2|G|}\sum_{i=\al,\be,\ga}\sum_{\lambda \in \Gamma_i}\sum_{\theta\in G_{\lambda_i}}\theta
-\frac{1}{2|G|}\sum_{i=\al,\be,\ga}|\Gamma_i|\,\id + \frac{1}{|G|}\,\id
\\&=\frac{1}{2}\sum_{i=\al,\be,\ga}\frac{\nu_i}{|G|}\sum_{\lambda \in \Gamma_i}\frac{1}{\nu_i}\sum_{\theta\in G_{\lambda_i}}\theta
-\frac{1}{2|G|}\left(\sum_{i=\al,\be,\ga}d_i-2\right)\id.
\end{align*}
We multiply this expression by $2$, and recognise the Reynolds operator for the stabiliser subgroups $\rey{G_{\lambda_i}}=\frac{1}{\nu_i}\sum_{\theta\in G_{\lambda_i}}\theta$. Moreover, we use the Riemann-Hurwitz formula (\ref{eq:finite subgroups of SO(3)}) rewritten to the form $\sum_{i\in \Omega}d_i-2=(|\Omega|-2)|G|$
and thus obtain
\[
2\rey{G}=\sum_{i=\al,\be,\ga}\frac{1}{d_i}\sum_{\lambda \in \Gamma_i}\rey{G_{\lambda}}
-\id.\]
To obtain the statement of the Lemma, one can now simply take the trace $\tr \rey{G}=\dim V^G$, keeping in mind that stabilizer subgroups of elements in one orbit are conjugate.
\end{proof}
\begin{Corollary}
If $V$ is a representation of a noncyclic polyhedral group (\ref{eq:presentation of G}) without trivial summand then \[V=V^{\langle \tal \rangle}\oplus V^{\langle\tbe\rangle}\oplus V^{\langle \tga\rangle}.\]
\end{Corollary}

\section{Rational Functions as a Quotient of a Graded Ring}
\label{sec:homogenisation}
In this section we describe rational functions on the projective line in terms of forms in two variables and subsequently express the rational functions with restricted poles as a quotient of a graded ring. This construction can be found in the literature, e.g.~\cite{eisenbud1995commutative, miranda1995algebraic}.
The advantage is that it gives the function ring of our interest an explicit $\text{SL}_2(\mb{C})$-module structure, which will prove to be fruitful in Section \ref{sec:squaring the ring}.

In order to describe meromorphic functions on the projective line $\mero$, notice that $\cp$ corresponds to the quotient of $\mb{C}^2\setminus\{0\}$ by the canonical action of $\mb{C}^*$. Thus there is a one to one relation between functions on $\cp$ and functions on $\mb{C}^2\setminus\{0\}$ which are $\mb{C}^*$-invariant. Examples of such functions are quotients of forms on $\mb{C}^2$ of identical degree (where we define a form as a multivariate polynomial where each term has the same degree; the degree of the form).
It is well known that this function field in fact corresponds to the field of meromorphic functions on $\cp$ \cite{miranda1995algebraic}. That is, 
\begin{equation}
\label{eq:mero quotients of forms}
\mero=\left\{[X:Y]\mapsto [P(X,Y):Q(X,Y)]\;|\;P,Q\in(\mb{C}[X,Y])_d,d\in\mb{N}_0\right\}.
\end{equation}

The polyhedral groups acting on $\cp$ have to be replaced by different groups when we translate the action to $\mb{C}^2$.
A linear action of $G$ on $\cp=\left(\mb{C}^2\setminus \{0\}\right)/\mb{C}^\ast$ is obtained by a projective representation $\rho:G\rightarrow \PGL(\mb{C}^2)$. Irreducible such representations can be constructed from irreducible linear representations $H \rightarrow \GL(\mb{C}^2)$ of central extensions $H\twoheadrightarrow G$ by making the diagram
\begin{center}
\begin{tikzcd}
H \arrow[r]\arrow[d,two heads]&\GL(\mb{C}^2)\arrow[d,two heads]\\
G \arrow[r]&\PGL(\mb{C}^2)
\end{tikzcd}
\end{center}
commute. That is, $\rho$ is obtained by choosing any section $G \rightarrow H$ and composing with the other maps. The central extension $H\twoheadrightarrow G$ is called sufficient if all projective representations of $G$ are obtained in this way.

In this paper we choose the binary polyhedral groups $H=G^\flat$ as sufficient extensions. They are defined as the preimages of the polyhedral groups $G\subset \PSL(\mb{C}^2)$ under the quotient map $\SL(\mb{C}^2)\rightarrow\PSL(\mb{C}^2)$. 
The binary polyhedral groups allow the presentation
\[G^\flat=\langle \tal, \tbe, \tga\;|\; \tal^{\nal}=\tbe^{\nbe}=\tga^{\nga}=\tal \tbe \tga\rangle\]
where the numbers $\nu_i$ correspond to the associated polyhedral group $G$. We choose to use the same symbols $\ti$ for generators of different groups, $G$ and $G^\flat$. The context should prevent confusion. 

\begin{Definition}[Ground form \cite{dolgachev2009mckay,klein1956lectures,klein1993vorlesungen,toth2002finite}]
\label{def:ground form}
Let $G<\Aut{\cp}$ and $\Gamma\in\bigslant{\cp}{G}$. Define $F_\Gamma\in\mb{C}[X,Y]$, up to $\mb{C}^*$, as the polynomial with divisor $D_\Gamma=\sum_{\gamma\in\Gamma}\gamma$, i.e. 
\[F_\Gamma(X,Y)=\prod_{[\gamma_x:\gamma_y]\in\Gamma}(\gamma_yX-\gamma_xY).\] Define $\nu_\Gamma=|G||\Gamma|^{-1}$. In the case of an exceptional orbit $\Gamma=\Gamma_i$ we use the short hand notation $F_i=F_{\Gamma_i}$ and $\nu_i=\nu_{\Gamma_i}$, and $F_i$ is called a \emph{ground form}. In general the ground forms are relative invariants of $G^\flat$ (there is a homomorphism  $\chi:G^\flat\rightarrow \mb{C}^\ast$ such that $\theta\cdot F_i=\chi(\theta) F_i$) and $F_\Gamma^{\nu_\Gamma}$ is an invariant.  Of special interest are the automorphic functions \[\ii=\frac{F_i^{\nu_i}}{F_\Gamma^{\nu_\Gamma}}\in\mero^G_\Gamma,\] which can be seen as the meromorphic function on $\bigslant{\cp}{G}$ with divisor  
$\nu_i {\Gamma_i}-\nu_\Gamma  \Gamma$.
\end{Definition}
\begin{Remark}
The term ground form originates from classical invariant theory and describes a (set of) covariant form(s) such that all other covariant forms can be derived
from it by transvection. In the context of this paper this would mean that one restricts one's attention to those relative invariants of minimal order ($F_i$ such that $|\Gamma_i|\le |\Gamma_{j}|$ for all $j\in\Omega$).
\end{Remark}

At this stage we have enough information to give an explicit relation between the algebra $\mero_\Gamma$ we want to study and the graded ring $R$ we will use to do so.
Define \[R_j=(\mb{C}[X,Y])_{jn}.\] That is, $R_j$ is the vector space of binary forms of degree $jn$.
Here $n\in\mb{N}$ is fixed (and will later be replaced by $|G|$). 
We denote by $R$ the direct sum of $R_j$ as abelian groups \[R=\bigoplus_{j\in\mb{Z}}R_j,\quad R_jR_{j'}\subset R_{j+j'},\] which forms a $\mb{Z}$-graded ring. Let $F$ be an element of $R_1$. The localisation of $R$ with respect to the multiplicative set of monomials in $F$ will be denoted by $R[F^{-1}]$. This is another $\mb{Z}$-graded ring. 
The subring of zero-homogeneous elements $(R[F^{-1}])_0$ is the ring of our interest. 

\begin{Lemma}
\label{lem:meromorphic as localisation}
Let $\Gamma\subset \cp$ be the set of zeros of $F\in R_1$. Then the map \[(R[F^{-1}])_0\ni \frac{P(X,Y)}{F(X,Y)^d}\mapsto \left([X:Y]\mapsto [P(X,Y):F(X,Y)^d]\right)\in\mero_\Gamma\] is an isomorphism of rings.
\end{Lemma}
\begin{proof}It is clear that the map $(R[F^{-1}])_0\rightarrow \mero_\Gamma$ is a monomorphism. It is less obvious that it is surjective.
To show this, let $\left([X:Y]\mapsto [P(X,Y):Q(X,Y)]\right)\in\mero_\Gamma$. 
Since the poles of this map are contained in $\Gamma=\{[\gamma_x^1:\gamma_y^1],\ldots,[\gamma_x^k:\gamma_y^k]\}$ we know that $Q$ has the form \[Q(X,Y)=\prod_{i=1}^k(\gamma_y^iX- \gamma_x^iY)^{d_i}\]
where $d_i\in\mb{N}_0$. Define $d=\max\{d_1,\ldots,d_k\}$. Then $F(X,Y)^d=A(X,Y)Q(X,Y)$ for some form $A(X,Y)$ and 
\begin{align*}
[P(X,Y):Q(X,Y)]&=[A(X,Y)P(X,Y):A(X,Y)Q(X,Y)]
\\&=[A(X,Y)P(X,Y):F(X,Y)^d]
\end{align*} 
which establishes surjectivity.
\end{proof}
\begin{Remark}
A straightforward generalisation shows that \[\mero_{\Gamma_1\cup\ldots\cup\Gamma_p}\cong(R[F^{-1}_1,\ldots,F^{-1}_p])_0\] where $\Gamma_1,\ldots,\Gamma_p\subset \cp$ and the zeros of $F_i\in R_1$ correspond to $\Gamma_i$. 
\end{Remark}

Henceforth we will identify $\mero_{\Gamma_1\cup\ldots\cup\Gamma_p}$ with $(R[F_1^{-1},\ldots,F_p^{-1}])_0$ without mention of the isomorphism, and we denote this ring by $S$: \[S=(R[F_1^{-1},\ldots,F_p^{-1}])_0\cong \mero_{\Gamma_1\cup\ldots\cup\Gamma_p}.\]

Now we incorporate an action of a finite group $G$. For any $G$-module $M$, the projection onto the isotypical summand $M^\chi$ associated to the irreducible character $\chi$ has the form $\pi_\chi=\frac{\chi(1)}{|G|}\sum_{\theta\in G}\overline{\chi(\theta)}\theta$. We take $G$ to be a polyhedral group acting on $\cp$, and therefore on $\mero$ by precomposition, and on $\mero_{\Gamma_1\cup\ldots\cup\Gamma_p}$ if and only if $\Gamma_1\cup\ldots\cup\Gamma_p$ is a union of orbits, in which case we may assume each $\Gamma_j$ is an orbit, and each $F_j$ is an invariant form.
The isotypical components of rational functions $\mero_{\Gamma_1\cup\ldots\cup\Gamma_p}^\chi$ that we study are isomorphic to $(R[F^{-1}_1,\ldots,F^{-1}_p])_0^\chi$ by Lemma \ref{lem:meromorphic as localisation}.  Since $G$ preserves the grading of $R$ we have $\pi_\chi(R[F^{-1}_1,\ldots,F^{-1}_p])_0=(\pi_\chi R[F^{-1}_1,\ldots,F^{-1}_p])_0$, and since $F^{-1}_1,\ldots,F^{-1}_p$ are $G$-invariant we have  $(\pi_\chi R[F^{-1}_1,\ldots,F^{-1}_p])_0=((\pi_\chi R)[F^{-1}_1,\ldots,F^{-1}_p])_0$. Thus we obtain 
\begin{equation}
\label{eq:isotypical components as localisation}
S^\chi=(R^\chi [F^{-1}_1,\ldots,F^{-1}_p])_0 \cong \mero_{\Gamma_1\cup\ldots\cup\Gamma_p}^\chi
\end{equation} and we study $R^\chi$ in Section \ref{sec:squaring the ring}.

We end this section with a convenient manner to obtain an isotypical component of rational functions from such a component with fewer poles. First notice that $S^\chi$ is a $S^G$-module. If $\Gamma$ and $\Gamma'$ are two $G$-invariant sets then 
\begin{equation}
\label{eq:more poles using invariants}
S^\chi_{\Gamma'}=S^G_{\Gamma'}S^\chi_{\Gamma}
\end{equation}
if and only if $\Gamma\subset\Gamma'$ (here we have temporarily shown the set of poles in the subscript).

\section{Squaring the Ring}
\label{sec:squaring the ring}
In this section we consider the polynomial ring $\mb{C}[X,Y]=\mb{C}[U]$ where $U$ is the natural representation of a binary polyhedral group $G^\flat$. Through elementary arguments we find the character of the subspaces $\mb{C}[U]_{|G|}$ of fixed degree $|G|=\frac{1}{2}|G^\flat|$, and we show how the character of higher degree subspaces $\mb{C}[U]_{h}$ can easily be recovered from the character of degrees below $|G|$, reducing the problem of describing general $\mb{C}[U]_{h}$ to a finite and straightforward computation of $\mb{C}[U]_{1},\ldots,\mb{C}[U]_{|G|-1}$. 
This will result in a description of the space of invariant vectors as a module over the ring of automorphic functions.

Let  $\|G\|$ denote the exponent of the group $G$, i.e. the least common multiple of the orders of group elements.
\begin{Lemma}
\label{lem:exponent}
Let $G^\flat$ be a binary polyhedral group corresponding to the polyhedral group $G$. Then
\[\|G^\flat\|=2\lcm(\nal, \nbe, \nga)=2\|G\|=\frac{2|G|}{|M(G)|}=\left\{\begin{array}{ll}|G|&\text{if }G\in\left\{\dg{2M},\tg,\og,\yg\right\}\\2|G|&\text{if }G\in\left\{\cg{N},\dg{2M+1}\right\}\end{array}\right.\]
where $M(G)$ is the Schur multiplier of $G$ and $\lcm$ stands for \emph{least common multiple}.
\end{Lemma}
\begin{proof}
We prove that $\|G^\flat\|=2\lcm(\nal, \nbe, \nga)$. The other equalities can be found in Table $\ref{tab:various properties of polyhedral groups}$. The polyhedral group
$G$ is covered by stabiliser subgroups $G_\lambda$, $\lambda\in\cp$. Thus, if $\pi:G^\flat\rightarrow G$ is any extension with kernel $Z$,  then $G^\flat$ is covered by the preimages $\pi^{-1}G_\lambda$, $\lambda\in\cp$. The order of an element $h\in G^\flat$ thus divides the order of the subgroup $\pi^{-1}G_\lambda$ containing it, which is $|Z|\nu_i$. Hence the exponent divides the least common multiple of these;
\[\|G^\flat\|\;|\;\lcm(|Z|\nal, |Z|\nbe, |Z|\nga)=|Z|\lcm(\nal, \nbe, \nga).\]
On the other hand, if we assume that $G^\flat$ is the binary polyhedral group then it is clear from the presentation $G^\flat=\langle \tal, \tbe, \tga\;|\; \tal^{\nal}=\tbe^{\nbe}=\tga^{\nga}=\tal \tbe \tga\rangle$ and the fact that $\tal\tbe\tga\in Z(G^\flat)=\cg{2}$ that the order of $\ti$ is $2\nu_i$, so that
\[2\lcm(\nal, \nbe, \nga)=\lcm(2\nal,2 \nbe, 2\nga)\;|\;\|G^\flat\|\] 
and $\|G^\flat\|=2\lcm(\nal, \nbe, \nga)$ as desired.
\end{proof}
Notice that if $G^\flat$ is a Schur cover, equal to the binary polyhedral group when possible, then $\|G^\flat\|=|G|$. For now we wish to use some results that are specific to $\SL_2(\mb{C})$.
Let $\chi$ be the character of the natural representation of $G^\flat$, i.e.~the monomorphism $\sigma:G^\flat\rightarrow \SL(\natrep)=\SL_2(\mb{C})$, and denote its symmetric tensor of degree $h$ by 
\newcommand{\sym}{\text{Sym}}
\[\chi_h=\chi_{\sym^hU}.\]
The \emph{Clebsch-Gordan decomposition} \cite{fossum1980invariant} for $\SL_2(\mb{C})$-modules
\begin{equation}
\label{eq:CG}
\natrep\otimes \sym^h\natrep=\sym^{h+1}\natrep\oplus \sym^{h-1}\natrep, \quad h\ge 2,
\end{equation} 
can be conveniently used to find the decomposition of $\chi_h$ when $h$ is a multiple of $\|G^\flat\|$.
To this end, we write the Clebsch-Gordan decomposition in terms of the characters 
\begin{equation}
\label{eq:CGmatrix}
\begin{pmatrix} \chi_h\\\chi_{h-1}\end{pmatrix}=
\begin{pmatrix}\chi&-1\\1&0\end{pmatrix}
\begin{pmatrix} \chi_{h-1}\\\chi_{h-2}\end{pmatrix}
\end{equation}
with boundary conditions
\[\chi_{-1}=0,\qquad \chi_0=\triv,\]
where $\triv$ is the trivial character.
The function $\chi_{-1}$ might not be defined by a symmetric tensor product but it gives a convenient boundary condition resulting in the correct solution, with $\chi_1=\chi$. 

\begin{Lemma}\label{lem:character degree reduction} 
If $\theta\in G^\flat$ has order $\nu>2$ and $\chi_h$ is the character of the order $h$ symmetric power of the natural representation of $G^\flat$, then \[\chi_{h+\nu}(\theta)=\chi_h(\theta)\] for all $h\in\mb{Z}$.
\end{Lemma}
\begin{proof}
Let $\zeta$ and $\zeta^{-1}$ be the eigenvalues of $\theta$'s representative in $\SL_2(\mb{C})$. In particular $\zeta^\nu=1$. The matrix which defines the linear recurrence relation (\ref{eq:CGmatrix}) becomes \[M(\theta)=\begin{pmatrix}\zeta+\zeta^{-1}&-1\\1&0\end{pmatrix}.\] We check that this matrix has eigenvalues $\zeta^{\pm 1}$ as well.
If the eigenvalues are distinct, i.e.~$\zeta^2  \ne 1$, i.e.~$\nu>2$, then $M(\theta)$ is similar to $\diag(\zeta,\zeta^{-1})$ and has order $\nu$, proving the lemma (notice that $M(\theta)$ is not diagonalisable if $\nu=1$ or $\nu=2$). 
\end{proof}

The only maps in  $\SL_2(\mb{C})$ whose square is the identity are $\pm\Id$.
The epimorphism $\pi$ in the short exact sequence
$1\rightarrow \cg{2}\rightarrow G^\flat \xrightarrow{\pi} G\rightarrow 1$ therefore satisfies
\begin{equation} 
\label{lem:order 2 in ker pi}
\theta\in \ker\pi \Leftrightarrow \theta^2=1.
\end{equation}
Moreover, if $1\ne z\in\ker\pi$ then its $\SL_2(\mb{C})$-representative is $-\Id$ and its action on a form $F\in \sym^hU\cong \sym^hU^\ast$ is given by 
\begin{equation}
\label{eq:central action on forms}
zF=(-1)^h F.
\end{equation}
Let $\regchar{G}$ and $\regchar{G^\flat}$ be the characters of the regular representations of $G$ and $G^\flat$ respectively, i.e.
\[
\pi^\ast \regchar{G}(\theta)=\left\{\begin{array}{ll} |G|&\pi(\theta)=1,\\[2mm] 0&\pi(\theta)\ne 1,\end{array}\right.\qquad
\regchar{G^\flat}(\theta)=\left\{\begin{array}{ll} 2|G|&\theta=1,\\[2mm] 0&\theta\ne 1,\end{array}\right.\]
where $\pi^\ast f(\theta)=f(\pi(\theta))$.

\begin{Theorem}
\label{thm:dimV}
Let $G$ be a polyhedral group and $G^\flat$ its corresponding binary polyhedral group, with epimorphism $\pi:G^\flat\rightarrow G$.
Let $\chi_h$ be the character of the $h$-th symmetric power of the natural representation of $G^\flat$, and $\regchar{G}$ and $\triv$ the character of the regular and trivial representation respectively. If $h,r\in\mb{N}_0$, $m\in\frac{2}{|M(G)|}\mb{N}_0$ and \[h=m|G|+r\] then
\[\chi_h=\left\{\begin{array}{ll}m\pi^*\regchar{G}+\chi_r&\text{if $h$ is even,}\\[2mm]m(\regchar{G^\flat}-\pi^*\regchar{G})+\chi_r&\text{if $h$ is odd.}\end{array}\right.\]
\end{Theorem}
\begin{proof} The statement is trivial for $m=0$. We prove the case when $h$ is even by evaluating both sides of the equation at each element of $G^\flat$. 

First of all, if $\pi(\theta)=1$ then $m \pi^\ast \regchar{G}(\theta)=m|G|$ and $\chi_h(\theta)=h+1$ and $\chi_r(\theta)=r+1$ because both $h$ and $r$ are even (cf. equation (\ref{eq:central action on forms})). 

If $\pi(\theta)\ne 1$ then $m \pi^\ast \regchar{G}(\theta)=0$. 
Also, by (\ref{lem:order 2 in ker pi}), $\theta$ has order $\nu>2$. Since $\nu\,|\,\|G^\flat\|=\frac{2|G|}{|M(G)|}\;|\;m|G|$ by Lemma \ref{lem:exponent}, one can apply Lemma \ref{lem:character degree reduction} to find $\chi_{h}(\theta)=\chi_r(\theta)$.
The second equality follows in a similar manner. 
\end{proof}

This result reduces the infinite problem of describing $\chi_h$, $h\in\mb{N}$, to the finite problem of describing $\chi_r$, $0\le r<|G|$. The particular symmetric power we will be most interested in is 
\begin{equation}
\label{eq:character of power |G|}
\chi_{|G|}= \pi^\ast \regchar{G}+\triv
\end{equation}
i.e. the case $m=1$, $r=0$, which is shown to be valid if $G\in\left\{\dg{2M},\tg,\og,\yg\right\}$.

\begin{Example}
As an illustration we compute the first twelve symmetric powers of the natural representation $\bt\hookrightarrow\SL(\mb{C}^2)$ of the binary tetrahedral group, using the Clebsch-Gordan decomposition (\ref{eq:CG}), and list them in Table \ref{tab:decomposition of symmetric powers} (possibly the easiest way to do this is using the affine Dynkin diagram $E_6$ in Figure \ref{fig:dynkinE6}, which is related to $\tg$ by the \emph{McKay correspondence}, cf.~\cite{dolgachev2009mckay,springer1987poincare}). Here we denote irreducible representation of $\bt$ by $\bt_i$, $i=1,\ldots,7$. This notation allows us later to compare representations from multiple polyhedral groups. Moreover, we drop the superscript $\flat$ and call the representation nonspinorial whenever the representation factors through $\tg$. Otherwise we say it is spinorial.
The number $\zeta_3\ne1$ is a cube root of unity. The natural representation $\bt\hookrightarrow\SL_2(\mb{C})$ is recognised as $\btiiii$ in the character table, since it is the only real valued $2$-dimensional character. We have underlined the character of a natural representation in the character tables.
\begin{center}
\begin{table}[ht!]
\caption{Irreducible characters of the binary tetrahedral group $\bt$.}
\label{tab:ctbt}
\begin{center}
\begin{tabular}{cccccccc} \hline
$\theta$& $1$ &$\tal^2$ &$\tga$&$z$&$\tbe^2$&$\tbe$&$\tal$\\
$|C_G(\theta)|$&$24$&$6$&$4$&$24$&$6$&$6$&$6$\\
\hline
$\bti$ & $1$ &$1$&$1$&$1$&$1$&$1$&$1$\\
$\btii$ & $1$ &$\zeta_3$&$1$&$1$&$\zeta_3^2$&$\zeta_3$&$\zeta_3^2$\\
$\btiii$  & $1$ &$\zeta_3^2$&$1$&$1$&$\zeta_3$&$\zeta_3^2$&$\zeta_3$\\
$\underline{\btiiii}$ & $2$ & $-1$ & $0$ & $-2$ & $-1$ & $1$ & $1$\\
$\btiiiii$ & $2$ & $-\zeta_3^2$ & $0$ & $-2$ & $-\zeta_3$ & $\zeta_3^2$ & $\zeta_3$\\
$\btiiiiii$ & $2$ & $-\zeta_3$ & $0$ & $-2$ & $-\zeta_3^2$ & $\zeta_3$ & $\zeta_3^2$\\
$\btiiiiiii$ & $3$ & $0$ & $-1$ & $3$ & $0$ & $0$ & $0$\\
\hline
\end{tabular}
\end{center}
\end{table}
\end{center}
\tikzstyle{dynkinnode}=[draw, color=black, shape=circle,minimum size=5.0 pt,inner sep=0]
\begin{center}
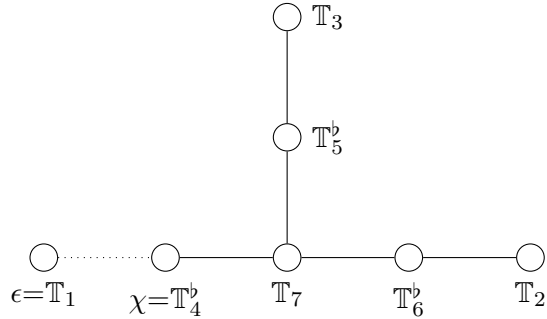
\begin{figure}[ht!]
\caption{Affine Dynkin diagram $E_6$ labelled by irreducible characters of $\bt$.}
\label{fig:dynkinE6}
\begin{center}
\begin{tikzpicture}[scale=.6]
  \path node at ( 0,0) [dynkinnode,label=270: $\triv{=}\bti$] (first) {$$}
  	node at ( 2,0) [dynkinnode,label=270: $\chi{=}\btiiii$] (second) {$$}
  	node at ( 4,0) [dynkinnode,label=270: $\btiiiiiii$] (third) {$$}
  	node at ( 6,0) [dynkinnode,label=270: $\btiiiiii$] (fourth) {}
  	node at (8,0) [dynkinnode,label=270: $\btii$] (fifth) {}
  	node at (4,1.3) [dynkinnode,label=0: $\btiiiii$] (sixth) {}
  	node at (4,2.6) [dynkinnode,label=0: $\btiii$] (seventh) {};
	\draw [dotted](first.east) to node [sloped,below]{} (second.west);
	\draw (second.east) to node [sloped,above]{} (third.west);
	\draw (third.east) to node [sloped,above]{} (fourth.west);
	\draw (fourth.east) to node [sloped,above]{} (fifth.west);
	\draw (third.north) to node [sloped,above]{} (sixth.south);
	\draw (sixth.north) to node [sloped,above]{} (seventh.south);
\label{dynkin}
\end{tikzpicture}
\end{center}
\end{figure}
\end{center}
Compare Theorem \ref{thm:dimV} for $G=\tg$ and $m=1$ (see $S^{11}\btiiii$ and $S^{12}\btiiii$ and note that $\pi^*\regchar{\tg}=\bti+\btii+\btiii+3\btiiiiiii$ and $\regchar{\bt}-\pi^*\regchar{\tg}=2\btiiii+2\btiiiii+2\btiiiiii$). 
Notice also that all even powers are nonspinorial and all odd powers are spinorial.
\begin{center}
\begin{table}[ht!]
\caption{Decomposition of symmetric powers $S^h\btiiii$, $h\le \|\bt\|=12$.}
\label{tab:decomposition of symmetric powers}
\begin{center}
\begin{tabular}{lll}\hline
$h$ & $S^h\btiiii$ & $\btiiii\otimes S^{h}\btiiii$  \\
\hline
$-1$ & $0$ & $0$\\
$0$ & $\bti$ & $\btiiii$ \\
$1$ & $\btiiii$ & $\btiiii\btiiii=\bti + \btiiiiiii$\\
$2$ & $\btiiiiiii$ & $\btiiii\btiiiiiii=\btiiii + \btiiiii+\btiiiiii$ \\
$3$ & $\btiiiii+\btiiiiii$ & $\btii+\btiii+2\btiiiiiii$ \\
$4$ & $\btii+\btiii+\btiiiiiii$ & $\btiiii+2\btiiiii+2\btiiiiii$\\
$5$ & $\btiiii+\btiiiii+\btiiiiii$ & $\bti+\btii+\btiii+3\btiiiiiii$\\
$6$ & $\bti+2\btiiiiiii$ & $3\btiiii+2\btiiiii+2\btiiiiii$\\
$7$ & $2\btiiii+\btiiiii+\btiiiiii$ & $2\bti+\btii+\btiii+4\btiiiiiii$\\
$8$ & $\bti+\btii+\btiii+2\btiiiiiii$ & $3\btiiii+3\btiiiii+3\btiiiiii$ \\
$9$ & $\btiiii+2\btiiiii+2\btiiiiii$ & $\bti+2\btii+2\btiii+5\btiiiiiii$ \\
$10$& $\btii+\btiii+3\btiiiiiii$ & $4\btiiiii+4\btiiiiii+3\btiiii$ \\
$11$& $2\btiiii+2\btiiiii+2\btiiiiii$ &  $2\bti+2\btii+2\btiii+6\btiiiiiii$\\
$12$& $2\bti+\btii+\btiii+3\btiiiiiii$ &\\\hline
\end{tabular}
\end{center}
\end{table}
\end{center}
Having computed the first twelve powers by hand, one can quickly find any other power using Theorem \ref{thm:dimV}. For instance, as $138=11\cdot 12+6$ we have $S^{138}\btiiii=11(\bti+\btii+\btiii+3\btiiiiiii)+\bti+2\btiiiiiii$.
\end{Example}

\section{The Dimension of the Module of Invariant Vectors}
In the previous section we used the Clebsch-Gordan decomposition to obtain the structure of $R^\chi$, where 
\[R=\bigoplus_{m\in\frac{2}{|M(G)|}\mb{N}_0} R_m,\quad R_m=\mb{C}[X,Y]_{m|G|},\] and $\{X,Y\}$ is a dual basis of a faithful representation $G^\flat\rightarrow\SL_2(\mb{C})$. In this section we map these findings to the isotypical components $S^\chi$ of the quotient $S\cong R/(F-1)$.

We have $\frac{2}{|M(G)|}=1$ for the most interesting and challenging groups, namely for $\tg,\;\og,\;\yg$ and $\dg{2M}$. This section is written for these polyhedral groups only, having the advantage that the factor $\frac{2}{|M(G)|}$ does not obscure the exposition. The main theorem at the end of this section can be obtained for the remaining groups ($\cg{N}$ and $\dg{2M+1}$) as well (cf.~\cite{knibbeler2014automorphic}). 

Theorem \ref{thm:dimV} shows that the character of $R_m$ is $m$ times the regular character of the polyhedral group, and thus, provides the following Poincar\'e series for the isotypical components $R^\chi$ of this $G$-module (for the case $|M(G)|=2$):
\begin{equation}
\label{eq:prehomogenised generating function}
P\left(R^\chi,t\right)=
\left\{
\begin{array}{ll}
0& \text{if $\chi$ is spinorial,}\\[2mm]
\frac{1}{(1-t)^2} & \text{if $\chi=\triv$,}\\[2mm]
\frac{\chi(1)^2 t}{(1-t)^2}  & \text{otherwise.}
\end{array}\right.
\end{equation}

It is well known that isotypical components of a polynomial ring are Cohen-Macaulay, with Krull dimension equal to the number of variables of the polynomial ring \cite{stanley1979invariants}. This means that there exists forms $\theta_1$ and $\theta_2$ and $\eta_1,\ldots,\eta_d$ such that $\mb{C}[X,Y]^\chi=\bigoplus_{r=1}^{d} \mb{C}[\theta_1,\theta_2]\eta_r$. From this fact, one can deduce that the component $R^\chi$, which is obtained from $\mb{C}[X,Y]^\chi$ by eliminating forms of degrees not in $|G|\mb{N}_0$, is Cohen-Macaulay as well. The Poincar\'e series (\ref{eq:prehomogenised generating function}) then prove the following result.
\begin{Proposition}
\label{prop:dimV generators}
Let $G^\flat$ be a binary polyhedral group, $\chi$ one of its irreducible characters and $U$ a natural $G^\flat$-module. Let $F_i$, $i\in\Omega$, be the ground forms in $\mb{C}[U]$, of respective degrees $\frac{|G|}{\nu_i}$. Then
\[R^{\chi}=
\left\{
\begin{array}{ll}
0& \text{if $\chi$ is spinorial,}\\[2mm]
\mb{C}[F_i^{\nu_i},F_{j}^{\nu_{j}}]& \text{if $\chi=\triv$ is trivial,}\\[2mm]
\bigoplus_{r=1}^{(\dim V)^2} \mb{C}[F_i^{\nu_i},F_{j}^{\nu_{j}}]P_r & \text{otherwise,}
\end{array}
\right.
\]
for some forms $P_1,\ldots,P_{(\dim V)^2}$ of degree $|G|$, and any $i, j\in\Omega$, $i\ne j$.
\end{Proposition}

Combining Proposition \ref{prop:dimV generators} with equation (\ref{eq:isotypical components as localisation}) allows us to describe the ring of automorphic functions and the dimension of the invariant vectors as a module over the ring of automorphic functions.

First some remarks on the invariant or automorphic functions $\mero_{\Gamma_1\cup\ldots\cup\Gamma_p}^G$. These are simply the holomorphic functions on the quotient surface $\cp/G$ with the points $\Gamma_1,\ldots\Gamma_p$ removed. The quotient $\cp/G$ is isomorphic to $\cp$. Hence the ring of automorphic functions is the ring of holomorphic functions on a $p$-punctured sphere.
The punctured sphere can be seen as an affine algebraic set, and the ring of automorphic functions as its coordinate ring. In particular, this ring is Noetherian. 
If $p=0,1,2$, the $p$-punctured sphere is isomorphic to the sphere, the plane and the cylinder respectively, and the rings of automorphic functions are $\mb{C}$, $\mb{C}[z]$ and $\mb{C}[z,z^{-1}]$ respectively. The case of $p=2$ is used in the concrete construction of affine Kac-Moody algebras.

We turn to the invariant vectors.
First we generalise the notation of Definition \ref{def:ground form}. 
\newcommand{\F}{\bar{F}}
As before, we write $F_\Gamma$ for the generator of the ideal $\mc{V}(\Gamma)$ in $\mb{C}[X,Y]$. If $\Gamma$ is an orbit of $G$ then the degree of $F_\Gamma$ divides $|G|$ and the power of  $F_\Gamma$ in $R_1$ will be denoted $\F_\Gamma$. 
\begin{Theorem}
\label{thm:dimV generators}
If $G<\Aut{\cp}$ is a finite group acting on a vector space $V$, then the space of invariant vectors \[\left(V\otimes\mero_{\Gamma_1\cup\ldots\cup\Gamma_p}\right)^G\] is a free $\mero_{\Gamma_1\cup\ldots\cup\Gamma_p}^G$-module whose dimension equals that of $V$.
\end{Theorem}
\begin{proof}
We consider the isotypical component $\mero_{\Gamma_1\cup\ldots\cup\Gamma_p}^\chi$ corresponding to an irreducible character $\chi$. Each product $V_{\overline{\chi}}\otimes V_\chi$ contains precisely one invariant up to scalar multiple, hence the dimension of the module $\mero_{\Gamma_1\cup\ldots\cup\Gamma_p}^\chi$ is $\chi(1)$ times the dimension of the module of invariant vectors $\left(V_{\overline{\chi}}\otimes\mero_{\Gamma_1\cup\ldots\cup\Gamma_p}\right)^G$.

Recall (\ref{eq:isotypical components as localisation}), i.e. $\mero_{\Gamma_1\cup\ldots\cup\Gamma_p}^\chi\cong(R^\chi[\F_1^{-1},\ldots,\F_p^{-1}])_0$. Incorporating Proposition \ref{prop:dimV generators} twice, we have 
\begin{align*}
\mero_{\Gamma_1\cup\ldots\cup\Gamma_p}^\chi&\cong((\bigoplus_{r=1}^{\chi(1)^2}\mb{C}[\F_\al,\F_\be]P_r)[\F_1^{-1},\ldots,\F_p^{-1}])_0
\\&=\bigoplus_{r=1}^{\chi(1)^2}(\mb{C}[\F_\al,\F_\be][\F_1^{-1},\ldots,\F_p^{-1}])_0\frac{P_r}{\F_1}
\\&\cong\bigoplus_{r=1}^{\chi(1)^2}\mero_{\Gamma_1\cup\ldots\cup\Gamma_p}^G \frac{P_r}{\F_1}
\end{align*}
and the statement follows at once.
\end{proof}
\begin{Corollary}
\label{cor:noetherian}
The space of invariant vectors as described in the previous theorem is a Noetherian module of the ring of automorphic functions.
\end{Corollary}
\begin{proof} By the previous theorem, the space of invariant vectors is finitely generated as a module over the ring of automorphic functions, and the latter is Noetherian.
\end{proof}
\begin{Remark}
The generators $\frac{P_r}{\F_1}$ are unique up to a transformation in $\GL(\chi(1)^2,\mero_{\Gamma_1\cup\ldots\cup\Gamma_p}^G)$. In particular, a diagonal transformation can be used to multiply $\frac{P_r}{\F_1}$ by a unit of the ring $\mero_{\Gamma_1\cup\ldots\cup\Gamma_p}^G$, for instance to get $\frac{P_r}{\F_{j_r}}$ for some $j_r\in\{1,\ldots,p\}$.
\end{Remark}
\begin{Remark}If $p=1$ then the Quillen-Suslin theorem 
(see for instance \cite{lang2002algebra}) can be used to prove that the invariant vectors form a free module over the ring of automorphic functions, since the latter is in this case the polynomial ring. However, this gives no information on its dimension.
\end{Remark}
Besides the value of the previous theorem to theoretical endeavours, for instance in the theory of Automorphic Lie Algebras as described below, there is a practical value for computational projects. Namely, in order to describe the infinite dimensional vector space $(V\otimes\mero_\Gamma)^G$ of invariant vectors, it is sufficient to find $\dim V$ independent invariant vectors in the finite dimensional vector space $V\otimes(\mb{C}[X,Y])_{|G|}$, vector valued forms of degree $|G|$.

\section{Divisor of Invariant Vectors}
\label{sec:divisor of invariant vectors}
In the previous section we showed that the space of invariant vectors in $V\otimes\mero_\Gamma$ is generated by $\dim V$ elements, but so far we have little information on these generators. In this section we obtain the zeros, with multiplicities, of their determinant. This provides a powerful invariant in the theory of Automorphic Lie Algebras.

We will make use of another, remarkably useful result.
\begin{Lemma}[\cite{neher2012irreducible}]
\label{lem:evaluating invariant vectors} For $\lambda\in\cp\setminus\Gamma$ define the evaluation map $\eval{\lambda}:V\otimes\mero_\Gamma\rightarrow V$ by $\eval{\lambda}(v\otimes f)=v\otimes f(\lambda)$. Then \[\eval{\lambda}\left(\left(V\otimes\mero_\Gamma\right)^G\right)=V^{G_\lambda}\] where $G_\lambda$ is the stabiliser subgroup $\{\theta\in G\,|\,\theta\lambda=\lambda\}$ of $G$.
\end{Lemma}
The inclusion $\eval{\lambda}\left(\left(V_{\overline{\chi}}\otimes\mero_\Gamma\right)^G\right)\subset V^{G_\lambda}$ is obvious, but the reverse inclusion is not. In \cite{knibbeler2014invariants} the author mistakenly presented Lemma \ref{lem:evaluating invariant vectors} as a new result, being unaware of the work in \cite{neher2012irreducible}, where in fact a more general statement is proved: 
the projective line can be replaced by any affine scheme.
 
The results of the previous section allow us to assign a positive divisor $\div{\chi}$ on the quotient space $\bigslant{\cp}{G}$ to each character $\chi$ of a polyhedral group, defined by the invariant vectors as follows.
Let $\chi$ be a nontrivial irreducible character of a polyhedral group $G$. By Proposition \ref{prop:dimV generators} there exist forms $P_{r,s}$, $r,s\in\{1,\ldots,\chi(1)\}$, freely generating $R^\chi$ over $R^G$, such that $\{P_{r,1},\ldots,P_{r,\chi(1)}\}$ is a basis of a representation of $G$ with character $\chi$, for each $r$. We can consider the determinant $\det(P_{r,s})$ of those generators, which transforms under the group $G$ as the determinant of the representation afforded by $\chi$. In particular, it is a relative invariant, and its set of zeros is $G$-invariant. Hence we can describe these zeros by a divisor $\div{\chi}$ on the space of orbits $\bigslant{\cp}{G}$.
In fact, by Lemma \ref{lem:evaluating invariant vectors}, $\det(P_{r,s})$ can only vanish on an exceptional orbit. 
\newcommand{\kapp}{\tilde{\kappa}}
This implies that it is a monomial $F_\al^{\kapp(\chi)_\al}F_\be^{\kapp(\chi)_\be}F_\ga^{\kapp(\chi)_\gamma}$ in ground forms, and the divisor takes the form 
\begin{equation}
\label{eq:determinant of invariant vectors tilde}
\div{\chi}=\kapp(\chi)_\al\Gal+\kapp(\chi)_\be\Gbe+\kapp(\chi)_\ga\Gga
\end{equation}
for some nonnegative integers $\kapp(\chi)_\al$, $\kapp(\chi)_\be$ and $\kapp(\chi)_\ga$.
The divisors of other characters of $G$ are defined by
\begin{align*}
&\div{\triv}=0,\\
&\div{\chi+\chi'}=\div{\chi}+\div{\chi'}.
\end{align*}
We will call $\div{\chi}$ the \emph{divisor of invariant vectors}.
Notice that, for an arbitrary character $\chi$, the degree of $\det(R_{r,s})$ equals $(\chi(1)-(\chi,\triv))|G|$.
The degrees of the equation $\det(R_{r,s})=F_\al^{\kapp(\chi)_\al}F_\be^{\kapp(\chi)_\be}F_\ga^{\kapp(\chi)_\gamma}$, using that $F_i$ has degree $\frac{|G|}{\nu_i}$, give us
\begin{equation}
\label{eq:sum orders by character}
\chi(1)-(\chi,\triv)=\frac{\kapp(\chi)_\al}{\nal}+\frac{\kapp(\chi)_\be}{\nbe}+\frac{\kapp(\chi)_\ga}{\nga}.\end{equation}

A general formula for the divisor of invariant vectors can be expressed using the following half integers.
\begin{Definition}[$\kappa(\chi)$]
\label{def:m(chi)}
Let $V$ be a module of a polyhedral group $G$ affording the character $\chi$. We define the half integer $\kappa(\chi)_i$ by
\[\kappa(\chi)_i=\nicefrac{1}{2}\;\codim V^{\langle \ti\rangle}=\frac{\chi(1)}{2}-\frac{1}{2\nu_i}\sum_{j=0}^{\nu_i-1}\chi(\ti^j),\qquad i\in\Omega,\]
where $\codim V^{\langle \ti\rangle}=\dim V-\dim V^{\langle \ti\rangle}$.
\end{Definition}
Lemma \ref{lem:dimV^g} for noncyclic polyhedral groups translates to 
\begin{equation}\label{eq:summ_i}\sum_{i=\al,\be,\ga}\kappa(\chi)_i=\dim V_\chi-\dim V^G_\chi=\chi(1)-(\chi,\epsilon).\end{equation}
Combining equation (\ref{eq:sum orders by character}) with equation (\ref{eq:summ_i}) gives us
\begin{equation}
\label{eq:sum orders}
\kappa(\chi)_\al+\kappa(\chi)_\be+\kappa(\chi)_\ga=\frac{\kapp(\chi)_\al}{\nal}+\frac{\kapp(\chi)_\be}{\nbe}+\frac{\kapp(\chi)_\ga}{\nga}.\end{equation}
We will prove below that in fact $\kapp(\chi)_i=\nu_i\kappa(\chi)_i$ if $\chi$ is real-valued. To this end, we first recall some general facts about group actions on Riemann surfaces.

Let $G$ be a finite group acting holomorphically and effectively on a Riemann surface $\rs$. Then a stabiliser subgroup $G_{\lambda_0}$, $\lambda_0\in \rs$, is cyclic. Moreover, if $G_{\lambda_0}=\langle \theta\rangle$ has order $\nu\ge 2$, then there is a local coordinate $z$ centred at $\lambda_0$ in which $\theta(z)=\zeta z$, where $\zeta$ is a primitive root of unity of order $\nu$ \cite{miranda1995algebraic}. As a consequence we have the following.
\begin{Lemma}\label{lem:order of semi invariant}
Adopting the above notation, if $f\in\mero[\rs]$ and $f(\theta^{-1}\lambda)=\zeta^jf(\lambda)$ for all $\lambda\in \rs$, then the Laurent expansion of $f$ in the local coordinate $z$ linearising $\theta$ is of the form \[f(z)=\sum_{n\in j+\mb{Z}\nu}f_n z^n.\] In particular, $\ord{\lambda_0}{f}\in j+\mb{Z}\nu$.
\end{Lemma}
\begin{proof}
In terms of the Laurent expansion, the hypothesis regarding the action of $\theta$ on $f(z)=\sum_{n\in\mb{Z}}f_n z^n$ reads $\sum_{n\in\mb{Z}}f_n (\zeta z)^n=\zeta^j \sum_{n\in\mb{Z}}f_n z^n$, i.e. \[\zeta^n f_n=\zeta^j f_n,\quad n\in\mb{Z}.\] This implies $f_n=0$ if $n\notin j+\mb{Z}\nu$, since $\zeta$ is a primitive root of unity of order $\nu$.
\end{proof}
Now we can prove the main result of this section.
\begin{Theorem}
\label{thm:determinant of invariant vectors}
The divisor of invariant vectors $\div{\chi}$ associated to a real-valued character $\chi$ of a noncyclic polyhedral group is given by 
\[\div{\chi}=\nal\kappa(\chi)_\al\Gal+\nbe\kappa(\chi)_\be\Gbe+\nga\kappa(\chi)_\ga\Gga.\]
\end{Theorem}
\begin{proof}
In light of (\ref{eq:determinant of invariant vectors tilde}) it remains to be shown that $\kapp(\chi)_i=\nu_i\kappa(\chi)_i$.

Let $[X_0:Y_0]\in\Gamma_i\subset\cp$ be fixed by a group element $\theta$ of order $\nu_i$. Let $\zeta$ be a root of unity such that $\theta z=\zeta z$ for a local coordinate $z$ centered at $[X_0:Y_0]$. Then $\zeta$ is a primitive root of unity of order $\nu_i$ because $G$ acts effectively on $\cp$.

If $\tau:G\rightarrow\GL(V)$ is the representation affording a real valued character $\chi$ then the eigenvalues of $\tau(\theta)$ are powers of the root of unity $\zeta$, since it is primitive, and the nonreal eigenvalues come in conjugate pairs. We denote the multiplicity of the eigenvalue $1$ and $-1$ by $\kappa_+$ and $\kappa_{-}$ respectively and the number of conjugate pairs by $\kappa_c$. That is, $\kappa_+=\dim V^{\langle \theta\rangle}$, $\kappa_++\kappa_-=\dim V^{\langle \theta^2\rangle}$ and $\kappa_++\kappa_-+2 \kappa_c=\chi(1)$. In particular, $\nicefrac{1}{2}\,\kappa_-+\kappa_c=\kappa(\chi)_i$.

Let the generators $P_{r,s}$ of $R^\chi$ be chosen such that $P_{r,1},\ldots,P_{r,\chi(1)}$ is a basis of a representation with character $\chi$ which diagonalises the group element $\theta$. If $\theta P_{r,s}=\zeta^j P_{r,s}$, where $0\le j\le \nu-1$, then $P_{r,s}$ has a zero of order at least $j$, by Lemma \ref{lem:order of semi invariant}.
We can assume that $P_{1,s},\ldots,P_{\chi(1),s}$ all transform as $\theta P_{r,s}=\zeta^j P_{r,s}$ for a fixed $j$, so that these forms, which make up a column in the matrix $(P_{r,s})$, contribute at least $j\Gamma_i$ to the divisor of $\det(P_{r,s})$.

Summing over all eigenvalues $\zeta^{\nu/2}=-1$ (in case $\nu$ is even) gives a contribution of $\nicefrac{1}{2}\,\nu_i\kappa_-\Gamma_i$ to the divisor, and summing over all pairs of conjugate eigenvalues $\zeta^{j}$ and $\zeta^{\nu_i-j}$ gives a contribution of $(j+\nu_i-j)\kappa_{c}\Gamma_i$ to the divisor. Thus the coefficient of $\Gamma_i$ in $\div{\chi}$ is bounded below by $\nu_i(\nicefrac{1}{2}\,\kappa_-+\kappa_c)=\nu_i\kappa(\chi)_i$. Using equation (\ref{eq:sum orders}) we see that this lower bound is sharp.
\end{proof}
In Appendix \ref{sec:examples} we compute various explicit examples illustrating this theorem.

\section{Automorphic Lie Algebras by Cocycles}
\label{sec:alias by cocycles}
Let $G$ be a finite group of automorphisms of $\cp$ and assume $G$ also acts on a complex Lie algebra $\mf{g}$ by Lie algebra automorphisms. Then, for any $G$-stable set $\Gamma\subset\cp$, $G$ acts on the Lie algebra $\mf{g}\otimes\mero_\Gamma$ by Lie algebra automorphisms (the Lie bracket of $\mf{g}\otimes\mero_\Gamma$ is taken to be the $\mero$-linear extension of the bracket of $\mf{g}$). The Lie subalgebra $\alia$ consisting of point wise fixed elements is called an Automorphic Lie Algebra \cite{Lombardo,lombardo2005reduction}
\[\alia=\left(\mf{g}\otimes\mero_\Gamma\right)^G\]
and the study of their Lie algebraic structure forms an active area of research. Automorphic Lie Algebras are also known by the name Equivariant Map Algebras. Indeed, $\alia$ consists of all $G$-equivariant meromorphic maps $\cp\rightarrow\mf{g}$ with poles confined to $\Gamma$.
In this section we describe the implications of our investigation of invariant vectors for the  theory of Automorphic Lie Algebras.
We simplify notation by introducing the letter
\[T=\mero_\Gamma^G\cong S^G\] for the ring of automorphic functions. 

As a first result, Theorem \ref{thm:dimV generators} describes the structure of an Automorphic Lie Algebra as a module over the automorphic functions. Explicitly, if $\dim_{\mb{C}}\mf{g}=k$, then there exists elements $m_1,\ldots,m_k$ such that 
\begin{equation}
\label{eq:alia dimension k}
\alia=Tm_1\oplus\ldots\oplus T m_k.
\end{equation}
This was also shown in \cite{chopp2011lie} using different techniques, and it appears that it was known by authors of earlier works on Automorphic Lie Algebras, but we were not able to find a proof in this literature.

It is common to say $\alia$ has rank $k$ in case of (\ref{eq:alia dimension k}), but we will say it has dimension $k$, to avoid confusion with the rank of the Lie algebra. We recall a fundamental fact.
\begin{Lemma}[{\cite[Corolary 4.4 (b)]{eisenbud1995commutative}}]
\label{lem:generators form free basis}
Let $T$ be a ring and $M$ a free $T$-module of dimension $k$. Then any set of $k$ elements of $M$ that generate $M$ forms a free basis; in particular, the dimension $k$ of $M$ is well defined.
\end{Lemma}

The evaluation map provides a useful tool in the study of Automorphic Lie Algebras. Lemma \ref{lem:evaluating invariant vectors} provides a surjective morphism
\begin{equation}\label{eq:evaluation epimorphism}
\eval{{\lambda_0}}:\alia\twoheadrightarrow \mf{g}^{G_{\lambda_0}},\quad \lambda_0\in\cp
\end{equation}
for all $\lambda_0\in\cp\setminus\Gamma$. In other words, $\alia$ is an extension of $\mf{g}^{G_{\lambda_0}}$ by the ideal $\ker \eval{{\lambda_0}}$, for each $\lambda_0$.
The Lie algebra $\mf{g}^{G_{\lambda_0}}$ is easily computed and belongs to a well-studied class of Lie algebras. Indeed, any stabiliser subgroup $G_{\lambda_0}$ of a finite group acting on a Riemann surface is a cyclic group. Thus we see that $\mf{g}^{G_{\lambda_0}}$ is nothing but $\mf{g}^{\theta}$ where $\theta$ is a torsion of $\mf{g}$. 

In this section we study Automorphic Lie Algebras whose base Lie algebra $\mf{g}$ is simple.
Inner torsions of simple Lie algebras are classified by Cartan \cite{cartan1927geometrie} and general torsions are classified by Kac \cite{kac1969automorphisms}. The fixed point algebras $\mf{g}^{\theta}$ are known as pseudo-Levi subalgebras of $\mf{g}$ and connected Lie groups with Lie algebra of the form $\mf{g}^{\theta}$ are called pseudo-Levi subgroups \cite{sommers1998generalization}. 
These groups act on any $\theta$-eigenspace in $\mf{g}$ (by restricting the adjoint representation). Quotients of pseudo-Levi subgroups by the kernels of these actions are also known as ${\theta}$-groups \cite{vinberg1976the}. They appear in the study of nilpotent orbits. The special case where $\theta$ has order $2$ is used to classify real forms and symmetric spaces \cite{helgason1978differential}. 

The Chevalley normal form for Automorphic Lie Algebras, introduced in \cite{lombardo2010on} and generalised in \cite{knibbeler2017higher}, is a set of generators for the Lie algebra as a $T$-module with Lie brackets partially described by the Cartan matrix of $\mf{g}$. More specifically, the generators define a $T$-module direct sum $\mf{n}_-\oplus\mf{h}\oplus\mf{n}_+$ with $\mf{h}$ an abelian algebra acting diagonally on these generators with eigenvalues described by the Cartan matrix of $\mf{g}$. However, different from the classical situation, the brackets of elements in $\mf{n}_\pm$ cannot be defined by the Cartan matrix alone. The remaining necessary information can be conveniently captured by a $2$-cocycles on the root system of $\mf{g}$ (cf. \cite{knibbeler2019cohomology}). The study of these $2$-cocycles is significantly easier than the study of Automorphic Lie Algebras directly. Below we explain how the results of this paper put constraints on the cocycles which can describe Automorphic Lie Algebras.

Existence of the normal from described above relies on the existence of $\mf{h}$ which we will call the Cartan subalgebra, defined below.
We call an endomorphism $E$ of a finitely generated $T$-module $M$ diagonalisable if there exist $z_1,\ldots,z_k\in M$ such that $M=Tz_1+\ldots+Tz_k$ and $Ez_i=a_i z_i$ for some $a_i\in T$. In case $M$ is a free $T$-module of dimension $k$, this implies that $M=Tz_1\oplus\ldots\oplus Tz_k$ by Lemma \ref{lem:generators form free basis}.

\begin{Definition}
A subalgebra $\mf{h}$ of an Automorphic Lie Algebra $\mf{A}$ is called a Cartan subalgebra (CSA) if each element $h\in\mf{h}$ has the property that $\ad(h)$ is a diagonalisable $T$-module endomorphism of $\mf{A}$, and $\mf{h}$ equals its own centraliser in $\mf{A}$.

A CSA $\mf{h}$ is called a constant spectrum CSA if there exists a set of generators for $\mf{h}$ as a $T$-module such that for each generator $h$ the eigenvalues of the $T$-module endomorphism $\ad(h)$ are in $\mb{C}$.
\end{Definition}

\begin{Remark} The fact that we work over a ring rather than a field complicates the definition of the CSA.
In a classical treatment of complex finite dimensional semisimple Lie algebras (e.g.~\cite{humphreys1972introduction}) one may define a CSA as a maximal subalgebra of semisimple elements.
The diagonalisability of CSA elements then requires an algebraically closed field. We set this problem aside by assuming diagonalisability directly.

If one would then opt to define a CSA as a maximal subalgebra of diagonalisable elements, then it is not clear whether a CSA equals its own centraliser, nor whether it is a summand of $\mf{A}$ as a $T$-module. The classical proof relies on the Jordan-Chevalley decomposition for endomorphisms of a vector space. This decomposition is only available for vector spaces over a perfect field. We set this problem aside as well, by assuming a CSA to equal its own centraliser.

So, why do we insist on this diagonalisability and require the CSA to be a summand, apart from the obvious benefits for describing the structure of the Lie algebra? 
This is simply because the computational classification results to date (cf.~\cite{knibbeler2017higher}) show that the class of Automorphic Lie Algebras containing such a nice CSA contains many interesting and relevant examples. Precisely which Automorphic Lie Algebras belong to this class remains an open problem.
\end{Remark}
\begin{Definition}
An Automorphic Lie Algebra is called hereditary if it contains a constant spectrum CSA.
\end{Definition}
The following Lemma allows us to speak of the dimension of a CSA and the rank of a hereditary Automorphic Lie Algebra.
\begin{Lemma}
\label{lem:csa free}
A CSA $\mf{h}$ of an Automorphic Lie Algebra $\alia$
defines a direct sum decomposition of $\alia$ into $\mf{h}$ and a finite number of root spaces.
Each summand is a free module over the ring of automorphic functions if $\mf{h}$ is a constant spectrum CSA or if poles are confined to a single orbit.
\end{Lemma}
\begin{proof}
Let $\mf{h}$ be a CSA of an Automorphic Lie Algebra $\alia$. Its adjoint image is a set of commuting diagonalisable endomorphisms of $\alia$, hence they are  simultaneously diagonalisable. This establishes $\alia$ as a direct sum of root spaces
\[\alia=\bigoplus_{{\alpha}}\alia_{\alpha}\] where ${\alpha}$ ranges over a subset of $T$-linear maps $\mf{h}\rightarrow T$ (called roots) and $\alia_{{\alpha}}=\{x\in\alia\,|\,[h,x]=\alpha(h)x, \forall h\in\mf{h}\}$. Since $\alia$ is Noetherian (Corollary \ref{cor:noetherian}), the number of roots is finite (the ascending chain $\alia_\alpha\subset\alia_\alpha\cup\alia_\beta\subset\ldots$ eventually stabilises).
The centraliser $\alia_0=C_{\alia}(\mf{h})=\{x\in\alia\,|\,[x,h]=0, \forall h\in\mf{h}\}$ equals $\mf{h}$ by definition, hence $\mf{h}$ is a summand of $\mf{A}$. Since $\mf{A}$ is free, we moreover notice that all root spaces are projective modules.

If the $G$-invariant set of poles $\Gamma$ defining $\alia$ is a single $G$-orbit, then freeness of the root spaces $\alia_{\alpha}$ follows from the Quillen-Suslin theorem, because $T$ is in this case a polynomial ring. In a more general setting, $T$ is no longer polynomial. However, in the constant spectrum case one can still use the Quillen-Suslin theorem as follows.

The CSA $\mf{h}$ is finitely generated over $T$ because $\alia$ is Noetherian. Let $h_1,\ldots,h_N$ be such generators, and let $\alpha(h_i)\in\mb{C}$. There exists $h\in\sum\mb{C}h_i$ such that all values $\alpha(h)$ are distinct. Indeed, the subset of $\sum\mb{C}h_i$ for which this is not the case is a finite collection of hyperplanes $\ker ((\alpha-\beta):\sum\mb{C}h_i\rightarrow\mb{C})$, which cannot cover $\sum\mb{C}h_i$.
Fix one $G$-orbit $\Gamma_1$ in $\Gamma$ and denote by $\alia_1$ the Automorphic Lie Algebra with poles confined to $\Gamma_1$. Notice that $\alia_1\subset\alia$.  Then we claim that 
\[\alia_1=\bigoplus_{{\alpha}}(\alia_1\cap\alia_{{\alpha}}).\]
That the right hand side is contained in the left hand side is clear.  Thanks to the existence of $h$, the reverse inclusion follows from a standard argument (cf.~\cite[Proposition 1.5]{kac1990infinite}). We include it for completeness. Let $\sum x_\alpha \in \alia_1$ with $x_\alpha\in \alia_\alpha$. We need to show that $x_\alpha\in \alia_1$. Because $\ad(h)$ has constant spectrum, it maps $\alia_1$ into itself. Hence $\ad(h)^k\sum x_\alpha=\sum\alpha(h)^k x_\alpha$, for $k=0,1,2\ldots$, are all elements of $\alia_1$. But $(\alpha(h)^k)$ is an invertible Vandermonde matrix (when $\alpha$ ranges over all roots and $k$ from 0 to one less than the number of roots). The inverse of this matrix takes the elements $\sum\alpha(h)^k x_\alpha$ to $x_\alpha$, proving that $x_\alpha\in\alia_1$ as well.

We have established $\alia_1\cap\alia_{{\alpha}}$ as a summand of $\alia_1$ hence a projective module over a polynomial ring and therefore free by the Quillen-Suslin theorem. Moreover, the dimension of $\alia_1\cap\alia_{{\alpha}}$ is finite since $\alia_1$ is Noetherian. Say
\begin{equation}
\label{eq:h1 free}
\alia_1\cap\alia_{{\alpha}}=T_1x_1\oplus\ldots\oplus T_1x_N
\end{equation} where $T_1$ is the ring of automorphic functions whose poles are confined to $\Gamma_1$. By (\ref{eq:more poles using invariants}) we have $\alia=T\alia_1$ and in particular $\alia_{{\alpha}}=T(\alia_1\cap\alia_{{\alpha}})=Tx_1+\ldots+ Tx_N$. Suppose $f_1x_1+\ldots+f_Nx_N=0$ for some $f_i\in T$. There exists $f_0\in T_1$ with zeros confined to $\Gamma\setminus\Gamma_1$, such that the order of those zeros is at least the maximum of the orders of the poles of $f_1,\ldots,f_N$, at each orbit. That is, $f_0f_i\in T_1$. Then  $f_0f_1x_1+\ldots+f_0f_Nx_N=0$ implies $f_0f_i=0$ for $i=1,\ldots,N$, by (\ref{eq:h1 free}). But $f_0$ is a unit in $T$ hence $f_i=0$ for $i=1,\ldots,N$, and therefore $\alia_{{\alpha}}=Tx_1\oplus\ldots\oplus Tx_N$.
\end{proof}

The main result of this section is that the Lie structure of a hereditary Automorphic Lie Algebra is described by  a $2$-cocycle on the classical root system in the sense of groupoid cohomology. Let us make precise what we mean by that.
\begin{Definition}
Let $(\Lambda,\ast)$ be a groupoid. That is, the set $\Lambda$ is equipped with a partial function $\ast:\Lambda\times\Lambda\rightarrow \Lambda$ which satisfies the axioms of a group. The set of $2$-cocycles $Z^2(\Lambda,M)$ consists of all partial maps $\omega^2$ with the same domain as $\ast$ and as codomain an abelian group $M$, which satisfy the cocycle condition
\begin{equation}
\label{eq:cocycle condition}
\omega^2(\beta,\gamma)-
\omega^2(\alpha\ast\beta,\gamma)+
\omega^2(\alpha,\beta\ast\gamma)-
\omega^2(\alpha,\beta)=0.
\end{equation}
The subset of symmetric cocycles, those satisfying $\omega^2(\alpha,\beta)=\omega^2(\beta,\alpha)$, is denoted $Z^2_+(\Lambda,M)$.
\end{Definition}
The set of weights of a representation of a simple Lie algebra is an example of an abelian groupoid. This will help us describe Automorphic Lie Algebras.
\begin{Theorem}
\label{thm:alia normal form}
Let $\alia$ be a hereditary Automorphic Lie Algebra based on a simple Lie algebra $\mf{g}$ of rank $N$ with root system $\roots$.
Then there is a direct sum of $T$-modules
\begin{equation}
\label{eq:alias cartan weyl basis}
\alia={\mf{h}}\oplus\bigoplus_{{\alpha}\in\roots}\alia_{\alpha},\quad {\mf{h}}=T{h}_1\oplus\ldots\oplus T{h}_{N},\quad \alia_\alpha=T{x}_\alpha,
\end{equation}
with the following properties. Each evaluation of $\mf{h}$ is a CSA of $\mf{g}$. If $\roots$ is the root system of $\mf{g}$ relative to a generic evaluation $\eval{\lambda_0}\mf{h}$, then $[h_i,x_\alpha]=\tilde{\alpha}(h_i) x_\alpha$ where $\tilde{\alpha}=\alpha\circ\eval{\lambda_0}$.

The generators ${x}_\alpha$ can be chosen such that
\begin{equation}
\label{eq:bracket of root vectors}
[{x}_{{\alpha}},{x}_{{\beta}}]=\epsilon(\alpha,\beta)\mb{I}^{\omega^2(\alpha,\beta)} {x}_{{\alpha}+{\beta}}.
\end{equation}
Here $\epsilon$ is the well known antisymmetric $\mb{Z}$-valued map describing a Chevalley basis of $\mf{g}$, as described by Tits \cite{tits1966constantes}, $\omega^2=(\omega^2_\al,\omega^2_\be,\omega^2_\ga)$ is an element of $Z^2_+(\roots\cup\{0\},\mb{Z}^3)$ taking values in $\mb{Z}_{\ge 0}^3$,
the multi index notation $\mb{I}^{\omega^2(\alpha,\beta)}=\prod_{i\in\Omega}\ii^{\omega^2_i(\alpha,\beta)}$ is used, and finally, the automorphic functions $\ii$ can be chosen to be $\mb{I}_{i,j_i}=\F_i \F_{j_i}^{-1}$ for any orbit of poles $\Gamma_{j_i}$ in $\Gamma$. In particular, if an exceptional orbit $\Gamma_i$ appears in $\Gamma$, then one can choose $\ii=\mb{I}_{i,i}=1$. Consequently, if all exceptional orbits appear in $\Gamma$, then $\alia$ is isomorphic to the current algebra $\mf{g}\otimes T$.
\end{Theorem}
\begin{proof}
By Lemma \ref{lem:csa free} and the assumption that $\alia$ contains a constant spectrum CSA $\mf{h}$, there exists a direct sum $\alia=\bigoplus_\alpha \alia_\alpha$ diagonalising the adjoint action of $\mf{h}=\alia_0$, where each summand is a free $T$-module, say $\alia_\alpha=\bigoplus_{j=1}^{N_\alpha}Tx_{\alpha,j}$. 
Theorem \ref{thm:dimV generators} implies that $\sum_\alpha N_\alpha =\dim_{\mb{C}}\mf{g}$. Let $\lambda_0\in\cp$ have trivial isotropy. Lemma \ref{lem:evaluating invariant vectors}, together with the previous, implies that $N_\alpha=\dim_\mb{C}\eval{\lambda_0}\alia_\alpha$ and $\eval{\lambda_0}x_{\alpha,1},\ldots,\eval{\lambda_0}x_{\alpha,N_\alpha}$ are $\mb{C}$-linearly independent.

We first show that the evaluation $\eval{\lambda_0}\mf{h}$ is a CSA of $\mf{g}$ (and in particular $N_0=N$). Notice that $\eval{\lambda_0}\mf{h}$ is an abelian subalgebra of $\mf{g}$ consisting of diagonalisable elements. Therefore $C_{\mf{g}}(\eval{\lambda_0}\mf{h})$ contains a CSA of $\mf{g}$ containing $\eval{\lambda_0}\mf{h}$. We now show that $C_{\mf{g}}(\eval{\lambda_0}\mf{h})=\eval{\lambda_0}\mf{h}$.  Let $H$ be an element of $C_{\mf{g}}(\eval{\lambda_0}\mf{h})$. By the evaluation lemma we have $H=\sum_{\alpha,j} c_{\alpha,j}\eval{\lambda_0}x_{\alpha,j}$ for some complex coefficients $c_{\alpha,j}$. Define $h=\sum_{\alpha,j} c_{\alpha,j}x_{\alpha,j}\in\alia$. Then $[x_{0,i},h]=\sum_{\alpha,j} c_{\alpha,j}\alpha(x_{0,i})x_{\alpha,j}$ vanishes at $\lambda_0$, for all $i$. But because $\eval{\lambda_0}x_{\alpha,j}$ are $\mb{C}$-linearly independent, this implies that the constants $c_{\alpha,j}\alpha(x_{0,i})$ are zero for all $\alpha,j,i$. Therefore, for each nonzero coefficient $c_{\alpha',j'}\ne 0$ we have $\alpha'(x_{0,i})=0$ for all $i$, hence $\alpha'=0$. This implies that $H\in\eval{\lambda_0}\mf{h}$, so that $\eval{\lambda_0}\mf{h}$ equals its centraliser and is a CSA of $\mf{g}$.
To see that the exceptional evaluations of $\mf{h}$ are CSAs of $\mf{g}$ as well, note that the constant spectrum does not allow the dimension of $\eval{\lambda}\mf{h}$ to depend on $\lambda$.

The generically evaluated Automorphic Lie Algebra $\eval{\lambda_0}\alia$ is nothing but the base Lie algebra $\mf{g}$ (by Lemma \ref{lem:evaluating invariant vectors} or (\ref{eq:evaluation epimorphism})), and $\eval{\lambda_0}\alia=\eval{\lambda_0}{\mf{h}}\oplus\bigoplus_{{\alpha}\ne 0}\eval{\lambda_0}\alia_{{\alpha}}$ is a Cartan Weyl decomposition of $\mf{g}$ because $\eval{\lambda_0}\alia$ is a CSA of $\mf{g}$. The roots of $\alia$ are therefore in bijective correspondence to the roots $\roots$ of $\mf{g}$ relative to the CSA $\eval{\lambda_0}{\mf{h}}$. Moreover, since
$\sum_\alpha N_\alpha =\dim_{\mb{C}}\mf{g}=N+|\roots|$ we have $N_\alpha=1$ for $\alpha\ne 0$. Thus we obtain (\ref{eq:alias cartan weyl basis}). For a more explicit description of the roots, evaluate $[h_i,x_\alpha]=\tilde{\alpha}(h_i)x_{\alpha}$ at $\lambda_0$. Because $\eval{\lambda_0}x_\alpha\ne 0$ we find that the constant eigenvalue $\tilde{\alpha}(h_i)$ corresponds to $\alpha(\eval{\lambda_0}h_i)$.

The brackets $[\eval{\lambda_0}{x}_{{\alpha}},\eval{\lambda_0}{x}_{{\beta}}]=\epsilon(\alpha,\beta)\eval{\lambda_0}{x}_{{\alpha}+{\beta}}$ are the classical brackets of the simple Lie algebra $\mf{g}$ as described by Chevalley and Tits. Hence the brackets in the Automorphic Lie algebra have the form $[{x}_{{\alpha}},{x}_{{\beta}}]=\epsilon(\alpha,\beta)f(\alpha,\beta) {x}_{{\alpha}+{\beta}}$
where $f(\alpha,\beta)\in T$. For any value of $\lambda$ outside the exceptional orbits $\Gamma_i$ the function $f(\alpha,\beta)$ is nonzero, as $\alia$ evaluates to $\mf{g}$. Hence, if $\Gamma$ is decomposed into orbits as $\bigcup_{j\in\zn{p}}\Gamma_j$, then $f(\alpha,\beta)$ must be a monomial in $\mb{I}_{i,j}$, $i\in\{\al,\be,\ga\}$, $j\in\zn{p}$. We can scale  $\mb{I}_{i,j}$ with respect to the particular generic element $\lambda_0$ such that  $\mb{I}_{i,j}(\lambda_0)=1$. Then $f(\alpha,\beta)$ is a monomial in  $\mb{I}_{i,j}$ with coefficient $1$, 
and thus $f$ is in fact described by a $2$-form $\tilde{\omega}^2$ with values in $\mb{Z}^{3p}$ by $f(\alpha,\beta)=\prod_{i,j}\mb{I}_{i,j}^{\tilde{\omega}^2_{i,j}(\alpha,\beta)}$. The $2$-form satisfies the cocycle condition (\ref{eq:cocycle condition}) and is symmetric because the Lie bracket of $\alia$ satisfies the Jacobi identity and is antisymmetric (cf.~\cite{knibbeler2019cohomology}).  

It remains to be shown that, for each $i\in\{\al,\be,\ga\}$, we can eliminate all but one of the $p$ $2$-cocycles $\tilde{\omega}^2_{i,j}$, $j\in\zn{p}$. This boils down to a cohomology question which was answered in \cite{knibbeler2019cohomology}: $H^2_+(\roots\cup\{0\},\mb{Z})=0$. Indeed, the automorphic functions $\mb{I}_{j,j'}$, for $j,j'\in\zn{p}$, are units of $T$, so that diagonal transformations of the form 
\begin{equation}
\label{eq:diagonal transformation}
h_i\mapsto h_i,\quad e_\alpha\mapsto \mb{I}_{j,j'}^{\omega^1(\alpha)}e_\alpha
\end{equation}
are invertible, for any $\omega^1:\roots\rightarrow\mb{Z}$. The fact that $H^2_+(\roots\cup\{0\},\mb{Z})=0$ and $\tilde{\omega}^2_{i,j}$ is a $2$-cocycle implies that there exists a map $\omega^1$ such that $\omega^1(\beta)-\omega^1(\alpha+\beta)+\omega^1(\alpha)=\tilde{\omega}^2_{i,j}(\alpha,\beta)$. With this choice of $\omega^1$, the transformation (\ref{eq:diagonal transformation}) sends $\tilde{\omega}^2$ to a new cocycles $\tilde{\tilde{\omega}}^2$ with $\tilde{\tilde{\omega}}^2_{i,j}=0$ and $\tilde{\tilde{\omega}}^2_{i,j'}={\omega}^2_{i,j}+{\omega}^2_{i,j'}$. Thus, for each exceptional orbit of zeros, $i\in\{\al,\be,\ga\}$, one can choose an orbit of poles, $j_i\in\zn{p}$, and by repeating the described  process, obtain a basis of the Automorphic Lie Algebra which produces structure constants described by the $2$-cocycles ${\omega}^2_{i,j}$ with ${\omega}^2_{i,j}=0$ for $j\ne j_i$ and ${\omega}^2_{i,j_i}=\sum_{j\in\zn{p}}\tilde{\omega}^2_{i,j}$. Now put $\omega^2_i=\omega^2_{i,j_i}$ to obtain (\ref{eq:bracket of root vectors}).
\end{proof}
We have remarked that most of the computed Automorphic Lie Algebras are hereditary, but it is hard to determine precisely which Automorphic Lie Algebras are hereditary. The previous theorem does provide a class of Automorphic Lie Algebras which are not hereditary.
\begin{Corollary}
If an Automorphic Lie Algebra, defined by an embedding $\rho:G\hookrightarrow\Aut{\mf{g}}$, is hereditary, then $\rho(G)\subset\Int{\mf{g}}$. 
\end{Corollary}
\begin{proof}
If $\mf{g}$ is simple, $\theta\in\Aut{\mf{g}}$ and $\mf{g}^\theta$ contains a CSA for $\mf{g}$, then $\theta\in\Int{\mf{g}}$ \cite[Chapter IX, Proposition 3]{jacobson1979lie}. By Theorem \ref{thm:alia normal form}, each value $\eval{\lambda_0}\alia$ of a hereditary Automorphic Lie Algebra contains a CSA for $\mf{g}$. Moreover, $\eval{\lambda_0}\alia=\mf{g}^{\theta_0}$ (using (\ref{eq:evaluation epimorphism}) and $\rho(G_{\lambda_0})=\langle \theta_0\rangle$) hence $\theta_0 \in\Int{\mf{g}}$. Finally, recall that $\rho(G)$ is the union of the stabiliser subgroups $\langle \theta_0\rangle$.
\end{proof}

\begin{Lemma}
\label{lem:kernels}
Let $\alia=\mf{h}\oplus\bigoplus_{\alpha\in\roots} T x_\alpha$ be the a Cartan Weyl decomposition (\ref{eq:alias cartan weyl basis}) of a hereditary Automorphic Lie Algebra.
For $i\in\Omega$ we pick $\lambda_i\in\Gamma_i\subset\cp$ and define
\begin{equation}
\label{eq:roots i}
\roots_{i}=\{\alpha\in\roots\;|\;\eval{\lambda_i}{x}_{{\alpha}}\ne0\}.
\end{equation}
This set is well defined. It is the root system of the reductive Lie algebra $\mf{g}^{G_{\lambda_i}}$ with respect to the CSA $\eval{\lambda_i}{\mf{h}}$, and $\mf{g}^{G_{\lambda_i}}$ is conjugate to $\mf{g}^{\ti}$ by an element of the reduction group, where $\ti$ is the generator of the polyhedral group $G$ that fixes an element of $\Gamma_i$.

In particular, if $\alpha,\beta\in\roots_i$ and $\alpha+\beta\in\roots$, then $\alpha+\beta\in\roots_i$ and
$\omega^2_i(\alpha,\beta)=0$,
where $\omega^2$ is the cocycle in (\ref{eq:bracket of root vectors}).
\end{Lemma}
\begin{proof}
For $\lambda_0\in\cp\setminus\Gamma$ we define $\roots_{\lambda_0}=\{\alpha\in\roots\;|\;\eval{\lambda_0}\alia_{{\alpha}}\ne\{0\}\}$. This set of roots depends only on the orbit $G\lambda_0\subset\cp$ by equivariance of the Automorphic Lie Algebra. Indeed, for $g\in G$ we have $\eval{(g\lambda_0)}\alia_\alpha=g\cdot(\eval{\lambda_0}\alia_\alpha)$, thus $\eval{(g\lambda_0)}\alia_\alpha=\{0\}\Leftrightarrow \eval{\lambda_0}\alia_\alpha=\{0\}$. 
Moreover, $\roots_{\lambda_0}$ is the set of nonzero eigenvalues of the Lie algebra $\mf{g}^{G_{\lambda_0}}$ with respect to the CSA  $\eval{\lambda_0}{\mf{h}}$. It is well known that $G_{\lambda_0}$ is a cyclic group $\langle\theta\rangle$, conjugate to $\langle\ti\rangle$ if $\lambda_0\in\Gamma_i$, in which case $\mf{g}^{G_{\lambda_0}}$ is conjugate to $\mf{g}^\ti$.
The last part is obtained by evaluating (\ref{eq:bracket of root vectors}) in $\lambda_0\in\Gamma_i$.
\end{proof}
We turn to the consequences of Theorem \ref{thm:determinant of invariant vectors} for the structure of Automorphic Lie Algebras. 
\begin{Theorem}
\label{thm:killing form}
Let $\omega^2$ be a $2$-cocycle defining the structure of a hereditary Automorphic Lie Algebra through (\ref{eq:bracket of root vectors}) and let $\roots_i$ be the root system (\ref{eq:roots i}). Then
\begin{equation*}
\omega^2_i(\alpha,-\alpha)=
\left\{
\begin{array}{ll}
0& \text{if $\alpha\in\roots_i$,}\\[2mm]
1& \text{if $\alpha\notin\roots_i$.}
\end{array}\right.
\end{equation*}
\end{Theorem}
\begin{proof}
Let $\{{h}_j,\,{x}_\alpha\,|\,j=1,\ldots,N,\,\alpha\in\roots\}$ be a Cartan Weyl basis realising (\ref{eq:alias cartan weyl basis}).
The determinant of the vectors ${h}_j$ and ${x}_\alpha$ is an element of $T$ given by \[\kappa(\chi_{\mf{g}})_i=|\roots^+|-|\roots_i^+|,\] up to constants, as described in Theorem \ref{thm:determinant of invariant vectors} (notice that $\chi_{\mf{g}}$ is a real character since the group preserves a bilinear form, the Killing form of $\mf{g}$). Since $\eval{\lambda_0}{\mf{h}}$ has constant dimension, the contribution of ${h}_i$ to this determinant is a factor in $\mb{C}^\ast$. The nontrivial contribution comes from the weight vectors ${x}_\alpha$. 

To be more precise on this point, we first fix a point $\lambda\in\cp\setminus\Gamma$. Then $\eval{\lambda}{\mf{h}}$ is a CSA of $\mf{g}$, and $\mf{g}$ is a vector space direct sum of the image and the kernel of $\ad(\eval{\lambda}{\mf{h}})$.
The elements $\eval{\lambda}x_\alpha={x}_\alpha(\lambda)$ are in the image. In other words, there exists a basis ${h}_1(\lambda),\ldots,{h}_N(\lambda),$ $n_1(\lambda),\ldots,n_{\dim \mf{g}-N}(\lambda)$ such that
${x}_\alpha(\lambda)=\sum_{j=1}^{\dim\mf{g}-N}f_\alpha^j(\lambda)n_j(\lambda)$. In this basis it is clear that the determinant of the vectors ${h}_j(\lambda)$ and ${x}_\alpha(\lambda)$ is given by $\det(f_\alpha^j(\lambda))$. In particular, all the $\kappa(\chi_\mf{g})_i$ factors $\ii$ in this determinant are provided by the weight vectors  ${x}_\alpha(\lambda)$.

If $\alpha\in\roots\setminus\roots_{i}$ then $[{x}_\alpha,{x}_{-\alpha}]$ must have a factor $\ii$. Indeed, the generators ${h}_i$ of the CSA do not vanish at any point of $\cp$ and $[{x}_\alpha,{x}_{-\alpha}]\in{\mf{h}}$ does because $\alpha\not\in\roots_{i}$. But since we know the total number of factors $\ii$ equals $\kappa(\chi_{\mf{g}})_i=|\roots^+|-|\roots_i^+|$ we can conclude that $[{x}_\alpha,{x}_{-\alpha}]$ has precisely one factor $\ii$ and if $\alpha\in\roots_{i}$ then $[{x}_\alpha,{x}_{-\alpha}]$ has no factor $\ii$. This proves the statement.
\end{proof}
We end this section with three consequences of Theorem \ref{thm:killing form}, for the Killing form, the  abelianisation (or first homology), and for the remaining values of the cocycle $\omega^2$ defining the Automorphic Lie Algebra.

We define the Killing $K$ form of an Automorphic Lie Algebra $\alia$ in the usual way: \[K:\alia\times\alia\rightarrow T,\quad (x,y)\mapsto \tr\ad(x)\ad(y).\] 
Here $\tr$ stands for trace and we take the trace of endomorphisms of $\alia$ as a finite dimensional free $T$-module, not as an infinite dimensional vector space over $\mb{C}$.
Theorem \ref{thm:killing form} can be seen as a description of the Killing form of Automorphic Lie Algebras.
\begin{Corollary} Let $K$ be the Killing form of a hereditary Automorphic Lie Algebra.
The determinant of $K$ has a divisor of zeros on $\bigslant{\cp}{G}$ given by
\begin{equation}
\label{eq:determinant killing form}
\div{\det(K)}=
2\left(\nal\kappa(\chi_\mf{g})_\al\Gamma_\al+\nbe\kappa(\chi_\mf{g})_\be\Gamma_\be+\nga\kappa(\chi_\mf{g})_\ga\Gamma_\ga\right)
\end{equation} 
where $\chi_\mf{g}$ is the character of the action of $G$ on $\mf{g}$.
If poles are allowed in a single orbit $\Gamma=\Gamma_1$, then this divisor is invariant under a change of generators of $\alia$ as a $T$-module. If poles are allowed in $p\ge2$ orbits $\Gamma=\bigoplus_{j\in\zn{p}}\Gamma_j$, then it is invariant up to $\bigoplus_{j\in\zn{p}}2\mb{Z}\Gamma_j$. 
\end{Corollary}
\begin{proof}
The values $K(h_i,h_j)$ correspond to the simple situation of $\mf{g}$, and $\det{K(h_i,h_j)}\in\mb{C}^*$. We have $K(x_\alpha,x_\beta)=0$ if $\alpha+\beta\ne 0$ since then $\ad(x_\alpha)\ad(x_\beta)$ is nilpotent. Therefore $K(x_\alpha,x_{-\alpha})\in\mb{C}^*\mb{I}^{\omega^2(\alpha,-\alpha)}$, using nondegeneracy of $K$ evaluated at a generic point.  The factor $2$ appears in (\ref{eq:determinant killing form}) because $K(x_\alpha,x_{-\alpha})$ and $K(x_{-\alpha},x_{\alpha})$ are both factors of the determinant).

To see that the divisor (\ref{eq:determinant killing form}) is invariant under a change of basis, note that such a change of basis is given by a matrix $M\in\GL(\dim\mf{g},T)$, and a bilinear form given by a matrix $K$ is mapped to one with matrix $M^{-t}KM^{-1}$. A matrix is invertible if and only if its determinant is a unit. Therefore, the divisor of $\det(M^{-t}KM^{-1})=\det(M)^{-1}\det(K)\det(M)^{-1}$ equals the divisor of $\det(K)$ up to the square of a unit. The units of the ring of automorphic functions $T$ with poles restricted to $\Gamma=\bigcup_{j\in\zn{p}}\Gamma_j$ are the nonzero monomials in $\mb{I}_{j,j'}$, $j,j'\in\zn{p}$. If $p=1$ then these are simply the nonzero constants and the divisor of zeros of $\det(K)$ is invariant.
\end{proof}

\begin{Corollary}
Let $\alia$ be a hereditary Automorphic Lie Algebra with root system $\roots$ and let $\roots_i$, $i\in\Omega$, be the associated exceptional root systems (\ref{eq:roots i}). Then
\[\dim_{\mb{C}}\bigslant{\alia}{[\alia,\alia]}=\sum_{\Gamma_i \not\subset\Gamma}\rank\roots-\rank\roots_i,\]
i.e. the abelianisation of the Automorphic Lie Algebra equals the direct sum of the abelianisations of the Lie algebras $\mf{g}^{G_\lambda}$, where the sum ranges over orbits $G\lambda$ not in $\Gamma$. 
\end{Corollary}
\begin{proof}
From the action of the constant spectrum CSA it is clear that $[\alia,\alia]$ contains $Tx_\alpha$ for all $\alpha\in\roots$. To see which subspace of the CSA is contained in $[\alia,\alia]$ one can use (\ref{eq:bracket of root vectors}) and Theorem \ref{thm:killing form}.
\end{proof}
Theorem \ref{thm:killing form} and its corollaries show that the root subsystems $\roots_i$, which are easily obtained from the group action due to Lemma \ref{lem:kernels}, determine the structure of the Automorphic Lie Algebra to a large extent. They do not tell the whole story though. That is, $\omega^2$ cannot be determined from $\roots_i$, and the analysis of this paper is not sufficient to determine $\omega^2$ from the group action. What we do find, as a final results in this section, is that the options for $\omega^2$ are rather limited: the structure constants (\ref{eq:bracket of root vectors}) have no squares. In low dimensions this can fix $\omega^2$ completely, a phenomenon that can be seen as the origin of the isomorphism conjecture discussed in earlier literature on Automorphic Lie Algebras.
\begin{Theorem}
The $2$-cocycle $\omega^2$ in (\ref{eq:bracket of root vectors}) takes all its values in $\{0,1\}^3$. 
\end{Theorem}
\begin{proof} Let $\alpha,\beta,\alpha+\beta \in \roots$. Then the cocycle condition (\ref{eq:cocycle condition}) on $\omega^2$, with $\gamma=-\beta$, reads
\begin{align*}
0
&=\omega^2(\beta,-\beta)-
\omega^2(\alpha+\beta,-\beta)+
\omega^2(\alpha,0)-
\omega^2(\alpha,\beta)
\\&=\omega^2(\beta,-\beta)-
\omega^2(\alpha+\beta,-\beta)-\omega^2(\alpha,\beta).
\end{align*}
Here we see that if $\omega^2(\alpha,\beta)_i\ge 2$ and $\omega^2(\alpha+\beta,-\beta)_i\ge 0$ then $\omega^2(\beta,-\beta)_i\ge 2$ which contradicts Theorem (\ref{thm:killing form}).
\end{proof}

\section{Concluding Remarks}
In this paper we have proven the existence of a $2$-cocycle $\omega^2$ on the root system, which, together with a Chevalley basis, completely describes a hereditary Automorphic Lie Algebra. This result paves the way to deeper understanding of the algebraic structure of Automorphic Lie Algebras. 
It also prompts further interesting questions; on the one hand, one would like to know which Automorphic Lie Algebras are hereditary. Explicit computations have shown that this class contains enough examples to justify its study. The optimistic guess would be that all Automorphic Lie Algebras with inner reduction group are hereditary. This statement is supported by the observation that no counter example has been found so far.
A second open problem relates to the complete description of the $2$-cocycle $\omega^2$. In this paper we  obtain partial results in this direction. Related to this is the classification $2$-cocycle with values $0$ and $1$ only. In the smallest case this already explains the isomorphisms between Automorphic Lie Algebras.

It is interesting to note that $2$-cocycles on root systems have not been studied in Lie theory before, as far as the authors are aware, despite the rather simple definition. The work of Kac on automorphisms of finite order and twisted loop algebras is closely related. It is only because loop algebras have poles in both exceptional orbits of the cyclic reduction group that the structure of the $2$-cocycle, which is hiding there, gets lost.

One might wonder about a possible connection between the usual Lie algebra cohomology and the root system cohomology.
But the differences both on the microscopic level (additive versus multiplicative, different symmetric versus antisymmetric behaviour) as well as on the macroscopic level (the acyclicity of the second cohomology in the simple case is a promising similarity, only to find that the first cohomology is trivial for a simple Lie algebra, but nontrivial for the root system), led us to believe that these are distinct cohomology theories.

\appendix
\section{Examples illustrating Theorem \ref{thm:determinant of invariant vectors}}
\label{sec:examples}
The numbers $\kappa(\chi)_i$ appear in the theory of Automorphic Lie Algebras due to Theorem \ref{thm:determinant of invariant vectors}. We list them in Table \ref{tab:1/2codimV^g}. Notice that  $\nu_i \kappa(\chi)_i$, $i\in\Omega$, are integers if and only if $\chi$ is real valued. The $\kappa$'s found in the latest developments \cite{knibbeler2017higher} on Automorphic Lie Algebras can be constructed from the table, by collecting the irreducible characters involved.

Table \ref{tab:1/2codimV^g} is determined using the characters of the polyhedral groups. Therefore we will present these first. 
Consider the dihedral group $\dg{N}=\langle r, s\;|\;r^N=s^2=(rs)^2=1 \rangle$.
If $N$ is odd then $\dg{N}$ has two one-dimensional characters, $\chi_1$ and $\chi_2$. If $N$ is even there are two additional one-dimensional characters, $\chi_3$ and $\chi_4$. 
\begin{center}
\begin{table}[ht!]
\caption{One-dimensional characters of $\dg{N}$.}
\label{eq:D_N one-dimensional characters}
\begin{center}
\begin{tabular}{ccccc}
\hline
 $\theta$ & $\chi_1$ & $\chi_2$ & $\chi_3$ &$\chi_4$\\
\hline
$r$ & $1$ & $1$ & $-1$ & $-1$\\
$s$ & $1$ & $-1$ & $1$ & $-1$\\
\hline
\end{tabular}
\end{center}
\end{table}
\end{center}
The remaining irreducible characters are two-dimensional (indeed, there is a normal subgroup of index $2$: $\cg{N}$, \cite{serre1977linear}). We denote these characters by $\psi_j$, for $1\le j <\nicefrac{N}{2}$. They take the values
\begin{equation}
\label{eq:D_N two-dimensional characters}
\psi_j(r^i)=\zeta_N^{ji}+\zeta_N^{-ji},\qquad \psi_j(sr^i)=0\,,
\end{equation}
where $\zeta_N= e^{\frac{2\pi i}{N}}.$

The characters of the binary tetrahedral group are collected in Table \ref{tab:ctbt} above. Below we give the character tables for $\bo$ and $\by$. For this last table we define the golden section $\phi^+$ and its conjugate $\phi^-$ in $\mb{Q}(\sqrt{5})$
\[\phi^\pm= \frac{1\pm\sqrt{5}}{2}.\] 
\begin{center}
\begin{table}[ht!]
\caption{Irreducible characters of the binary octahedral group $\bo$.}
\label{tab:ctbo}
\begin{center}
\begin{tabular}{ccccccccccccc}\hline
$\theta$& $1$ & $\tga$ & $\tbe^2$ & $\tal^2$ & $z$ & $\tal^3$ & $\tbe$ & $\tal$ \\
$|C_G(\theta)|$&$48$&$4$&$6$&$8$&$48$&$8$&$6$&$8$\\
\hline
$\boi$ & $1$ &$1$&$1$&$1$&$1$&$1$&$1$&$1$\\
$\boii$ & $1$ & $-1$ & $1$ & $1$ & $1$ &$-1$&$1$&$-1$ \\
$\boiii$ & $2$ &$0$&$-1$&$2$&$2$&$0$&$-1$&$0$\\
$\underline{\boiiii}$ & $2$ &$0$&$-1$&$0$&$-2$&$-\sqrt{2}$&$1$&$\sqrt{2}$\\
$\boiiiii$ & $2$ &$0$&$-1$&$0$&$-2$&$\sqrt{2}$&$1$&$-\sqrt{2}$\\
$\boiiiiii$ & $3$ &$1$&$0$&$-1$&$3$&$-1$&$0$&$-1$\\
$\boiiiiiii$ & $3$ &$-1$&$0$&$-1$&$3$&$1$&$0$&$1$\\
$\boiiiiiiii$ & $4$ &$0$ & $1$ & $0$ & $-4$ & $0$ & $-1$ & $0$\\\hline
\end{tabular}
\end{center}
\end{table}
\end{center}

\begin{center}
\begin{table}[ht!]
\caption{Irreducible characters of the binary icosahedral group $\by$.}
\label{tab:ctby}
\begin{center}
\begin{tabular}{cccccccccccccc} \hline
$\theta$& $1$ & $\tal^2$ &$\tal^4$&$\tbe$&$\tga$&$\tbe^2$&$\tal^3$&$z$&$\tal$\\
$|C_G(\theta)|$&$120$&$10$&$10$&$6$&$4$&$6$&$10$&$120$&$10$\\
\hline
$\byi$ & $1$ &$1$&$1$&$1$&$1$&$1$&$1$&$1$&$1$\\
$\underline{\byii}$ & $2$ &$-\phi^-$&$-\phi^+$&$1$&$0$&$-1$&$\phi^-$&$-2$&$\phi^+$\\
$\byiii$ & $2$ &$-\phi^+$&$-\phi^-$&$1$&$0$&$-1$&$\phi^+$&$-2$&$\phi^-$\\
$\byiiii$ & $3$ &$\phi^+$&$\phi^-$&$0$&$-1$&$0$&$\phi^+$&$3$&$\phi^-$\\
$\byiiiii$ & $3$ &$\phi^-$&$\phi^+$&$0$&$-1$&$0$&$\phi^-$&$3$&$\phi^+$\\
$\byiiiiii$ & $4$&$-1$&$-1$&$1$&$0$&$1$&$-1$&$4$&$-1$\\
$\byiiiiiii$ & $4$&$-1$&$-1$&$-1$&$0$&$1$&$1$&$-4$&$1$\\
$\byiiiiiiii$ & $5$&$0$&$0$&$-1$&$1$&$-1$&$0$&$5$&$0$\\
$\byiiiiiiiii$ & $6$&$1$&$1$&$0$&$0$&$0$&$-1$&$-6$&$-1$\\\hline
\end{tabular}
\end{center}
\end{table}
\end{center}
\newcommand{\w}{3.5}
\begin{center}
\begin{table}[ht!]
\caption{Rows $3$, $4$ and $5$ display $\kappa(\chi)_i$ and rows $6$, $7$ and $8$ display $\nu_i \kappa(\chi)_i$.}
\label{tab:1/2codimV^g}
\begin{center}
\begin{tabular}{p{\w mm}p{\w mm}p{\w mm}p{\w mm}p{\w mm}p{\w mm}p{\w mm}p{\w mm}p{\w mm}p{\w mm}p{\w mm}p{\w mm}p{\w mm}p{\w mm}p{\w mm}p{\w mm}} \hline
$ $&$\chi_2$&$\chi_3$&$\chi_4$&$\psi_j$&$\btii$&$\btiii$&$\btiiiiiii$&$
\boii$&$\boiii$&$\boiiiiii$&$\boiiiiiii$&$\byiiii$&$\byiiiii$&$\byiiiiii$&$\byiiiiiiii$\\
\hline
 &$1$&$1$&$1$&$2$&$1$&$1$&$3$&$1$&$2$&$3$&$3$&$3$&$3$&$4$&$5$\\
\hline
$\al$&$0$&$\nicefrac{1}{2}$&$\nicefrac{1}{2}$&$1$&$\nicefrac{1}{2}$&$\nicefrac{1}{2}$&$1$&
$\nicefrac{1}{2}$&$\nicefrac{1}{2}$&$\nicefrac{3}{2}$&$1$&$1$&$1$&$2$&$2$\\
$\be$&$\nicefrac{1}{2}$&$\nicefrac{1}{2}$&$0$&$\nicefrac{1}{2}$&$\nicefrac{1}{2}$&$\nicefrac{1}{2}$&$1$
&$0$&$1$&$1$&$1$&$1$&$1$&$1$&$2$\\
$\ga$&$\nicefrac{1}{2}$&$0$&$\nicefrac{1}{2}$&$\nicefrac{1}{2}$&$0$&$0$&$1$&
$\nicefrac{1}{2}$&$\nicefrac{1}{2}$&$\nicefrac{1}{2}$&$1$&$1$&$1$&$1$&$1$\\
\hline 
$\al$&$0$&$\nicefrac{N}{2}$&$\nicefrac{N}{2}$&$N$&$\nicefrac{3}{2}$&$\nicefrac{3}{2}$&$3$&
$2$&$2$&$6$&$4$&$5$&$5$&$10$&$10$\\
$\be$&$1$&$1$&$0$&$1$&$\nicefrac{3}{2}$&$\nicefrac{3}{2}$&$3$
&$0$&$3$&$3$&$3$&$3$&$3$&$3$&$6$\\
$\ga$&$1$&$0$&$1$&$1$&$0$&$0$&$2$&
$1$&$1$&$1$&$2$&$2$&$2$&$2$&$2$\\
\hline 
\end{tabular}
\end{center}
\end{table}
\end{center}

\begin{Example}[Determinants of $\dg{N}$-invariant vectors]\label{ex:determinant of dihedral invariant vectors}
For the dihedral group we have explicit descriptions for the isotypical components $\mb{C}[X,Y]^\chi$ at hand; they are explicitly computed in both \cite{knibbeler2014automorphic} and \cite{knibbeler2014invariants}. Here we give the results without derivation.

The characters of $\dg{N}=\langle r, s\;|\;r^N=s^2=(rs)^2=1 \rangle$ are given above. We will choose bases in which $r$ acts diagonally, so that the matrices for the generators in the representation $\psi_j$ become
\begin{equation}
\label{eq:standard matrices}
 \rho(r)=\begin{pmatrix}\zeta_N^{{j}}&0\\0&\zeta_N^{N-{j}}\end{pmatrix} ,\qquad \rho(s)=\begin{pmatrix}0&1\\1&0\end{pmatrix},
\end{equation}
where $\zeta_N=e^{\frac{2\pi i}{N}}$.
By confirming that the group relations $\rho_r^N=\rho_s^2=(\rho_r\rho_s)^2=\Id$ hold and that the trace $\psi_j=\tr \circ\rho$ is given by (\ref{eq:D_N two-dimensional characters}) one can check that this is indeed the correct representation. 

Suppose $\{X,Y\}$ is a basis for $V_{\psi_1}^\ast$, corresponding to (\ref{eq:standard matrices}).
One then finds the relative invariant forms
\begin{equation}
\label{eq:Fs}
\fal=XY,\qquad\fbe=\frac{X^N+Y^N}{2},\qquad \fga=\frac{X^N-Y^N}{2},
\end{equation} 
and equivariant vectors as presented in Table \ref{tab:invariants, N odd} and Table \ref{tab:invariants, N even}. 
\begin{table}[ht!]
\caption{Module generators $\eta_i$ in $(V_\chi\otimes \mb{C}[X,Y])^{\dg{N}}=\bigoplus_i\mb{C}[\fal,\fbe]\eta_i$, $N$ odd.}
\label{tab:invariants, N odd}
\begin{center}
\begin{tabular}{cccc}\hline
$\chi$ & $\chi_1$ & $\chi_2$ & $\psi_j$\\
\hline
$\eta_i$ & $1$ & ${\fga}$ & $\begin{pmatrix}X^j\\Y^j\end{pmatrix},\;\begin{pmatrix}Y^{N-j}\\X^{N-j}\end{pmatrix}$\\\hline
\end{tabular}
\end{center}
\end{table}
\begin{table}[ht!]
\caption{Module generators $\eta_i$ in $(V_\chi\otimes \mb{C}[X,Y])^{\dg{2N}}=\bigoplus_i\mb{C}[\fal,\fbe^2]\eta_i$.}
\label{tab:invariants, N even}
\begin{center}
\begin{tabular}{cccccc}\hline
$\chi$ & $\chi_1$ & $\chi_2$ & $\chi_3$ & $\chi_4$ & $\psi_j$\\
\hline
$\eta_i$ & $1$ & $\fbe\fga$ & $\fbe$ & $\fga$ & $\begin{pmatrix}X^j\\Y^j\end{pmatrix},\;\begin{pmatrix}Y^{2N-j}\\X^{2N-j}\end{pmatrix}$\\\hline
\end{tabular}
\end{center}
\end{table}

Now we can readily check Theorem \ref{thm:determinant of invariant vectors} for $\dg{N}$ by computing all the determinants and compare in each case the exponents of $\fal$, $\fbe$ and $\fga$ to the relevant column in Table \ref{tab:1/2codimV^g}.  First notice that all representations of $\dg{N}$ are of real type.

We start with $\chi_2$. If $N$ is odd then $\mb{C}[X,Y]^{\chi_2}_{2N}=\left(\mb{C}[\fal,\fbe]\fga\right)_{2N}=\mb{C}\fbe\fga$. If $N$ is even, we use the extension $G^\flat=\dg{2N}$ and also find $\mb{C}[X,Y]^{\chi_2}_{2N}=\left(\mb{C}[\fal,\fbe^2]\fbe\fga\right)_{2N}=\mb{C}\fbe\fga$. The exponents $(0,1,1)$ are indeed identical to $(\nal \kappa(\chi_2)_\al,\, \nbe \kappa(\chi_2)_\be,\, \nga \kappa(\chi_2)_\ga)$ as found in the $\chi_2$-column of Table~\ref{tab:1/2codimV^g}.

For $\chi_3$ and $N$ even, $G^\flat=\dg{2N}$, one finds $\mb{C}[X,Y]^{\chi_3}_{2N}=\left(\mb{C}[\fal,\fbe^2]\fbe\right)_{2N}=\mb{C}\fal^{\frac{N}{2}}\fbe$, and the last linear character $\chi_4$ gives  $\mb{C}[X,Y]^{\chi_4}_{2N}=\left(\mb{C}[\fal,\fbe^2]\fga\right)_{2N}=\mb{C}\fal^{\frac{N}{2}}\fga$, also in agreement with Table \ref{tab:1/2codimV^g}.

Now we consider the two-dimensional representations, when $N$ is odd;
\begin{align*}
\det \mb{C}[X,Y]^{\psi_j}_{2N}&=\det\left(\mb{C}[\fal,\fbe]\begin{pmatrix}X^j&Y^j\\Y^{N-j}&X^{N-j}\end{pmatrix}\right)_{2N}\\[2mm]
&=\left\{
\begin{array}{ll}
\mb{C}\det\begin{pmatrix}\fal^{\frac{N-j}{2}}\fbe X^j&\fal^{\frac{N-j}{2}}\fbe Y^j\\\fal^{\frac{N+j}{2}} Y^{N-j}&\fal^{\frac{N+j}{2}}X^{N-j}\end{pmatrix}= \mb{C}\fal^N\fbe\fga& j\text{ odd,} \\[6mm]
\mb{C}\det\begin{pmatrix}\fal^{\frac{2N-j}{2}} X^j&\fal^{\frac{2N-j}{2}}Y^j\\\fal^{\frac{j}{2}}\fbe Y^{N-j}&\fal^{\frac{j}{2}}\fbe X^{N-j}\end{pmatrix}=\mb{C}\fal^N\fbe\fga& j\text{ even.}
\end{array}\right.
\end{align*}

If $N$ is even and $G^\flat=\dg{2N}$, we only consider $2j$ because the statement of Theorem \ref{thm:determinant of invariant vectors} only concerns representations of $G$, not for instance spinorial representation of $G^\flat$;
\begin{align*}
\det \mb{C}[X,Y]^{\psi_{2j}}_{2N}&=\det\left(\mb{C}[\fal,\fbe^2]\begin{pmatrix}X^{2j}&Y^{2j}\\Y^{2N-{2j}}&X^{2N-{2j}}\end{pmatrix}\right)_{2N}\\[2mm]
&=
\mb{C}\det\begin{pmatrix}\fal^{N-j} X^{2j}&\fal^{N-j} Y^j\\\fal^{j} Y^{2N-2j}&\fal^{j}X^{2N-2j}\end{pmatrix}=\mb{C}\fal^N (X^{2N}-Y^{2N})=\mb{C}\fal^N\fbe\fga.
\end{align*}
This confirms Theorem \ref{thm:determinant of invariant vectors} for all representations of the dihedral group.
\end{Example}

The other examples that are feasible to do by hand are the remaining one-dimensional characters. These include the only representations of the polyhedral groups that are not real valued, namely $\btii$ and $\btiii$, cf.~Table \ref{tab:ctbt}.
\begin{Example}[One-dimensional characters]\label{ex:determinant of invariant forms}
Let the characters of the ground forms be indexed by $\Omega$, \[\theta F_i=\chi_i(\theta)F_i,\qquad i\in\Omega.\] 
For the tetrahedral group, the degrees of the ground forms are $\dal=\dbe=4$ and $\dga=6$, cf.~Table \ref{tab:various properties of polyhedral groups}. Necessarily $\chi_\al=\btii$ and $\chi_\be=\btiii$, or the other way around. There is no way to distinguish, (but they cannot be equal because there is a degree 8 invariant, which can be deduced using the Molien series) 
so we take the first choice. The last ground form is invariant, $\chi_\ga=\bti$, because $\chi_\ga^\nga=\bti$,  $\nga=2$, and there is no element of order two in $\mc{A}\tg=\cg{3}$. We get $\mb{C}[X,Y]^{\bti}=\mb{C}[\fal\fbe,\fga](\fal^3+ \fbe^3)$ and $\mb{C}[X,Y]^{\btii}=\mb{C}[\fal\fbe,\fga](\fal\oplus \fbe^2)$ and $\mb{C}[X,Y]^{\btiii}=\mb{C}[\fal\fbe,\fga](\fal^2\oplus\fbe)$, using for instance generating functions. 
Thus 
\begin{align*}
&\mb{C}[X,Y]^{\btii}_{12}=\left(\mb{C}[\fal\fbe,\fga](\fal\oplus \fbe^2)\right)_{12}=\mb{C}\fal^2\fbe,\\
&\mb{C}[X,Y]^{\btiii}_{12}=\left(\mb{C}[\fal\fbe,\fga](\fal^2\oplus\fbe)\right)_{12}=\mb{C}\fal\fbe^2,\\
&\det \mb{C}[X,Y]^{\btii+\btiii}_{12}=\mb{C}[X,Y]^{\btii}_{12}\mb{C}[X,Y]^{\btiii}_{12}=\mb{C}\fal^3\fbe^3.
\end{align*}
This example illustrates the necessity to consider $\chi+\chi^*$ rather than simply $\chi$ in Theorem \ref{thm:determinant of invariant vectors}.

The remaining one-dimensional character to check is $\boii$. We have $\chi_\al=\boii$, $\chi_\be=\boi$ and $\chi_\ga=\boii$ and
$\mb{C}[X,Y]^{\boii}_{24}=\left(\mb{C}[\fal^2,\fbe](\fal\oplus\fga)\right)_{24}=\mb{C}\fal^2\fga$,
in correspondence with the $\boii$-column in Table \ref{tab:1/2codimV^g}.
\end{Example}

\bibliography{bibliography} 
\bibliographystyle{amsplain}

\end{document}